\documentclass[%
 reprint,prl
 amsmath,amssymb,
 aps,%prl
]{revtex4-2}

\usepackage [utf8]{inputenc}
\usepackage [english]{babel}
\usepackage{enumerate}
\usepackage{fullpage}
\usepackage{bbm}
\usepackage{physics}

\usepackage{graphicx}
\usepackage{amsmath,amsthm,amssymb,amsfonts,setspace}
\usepackage{bm}
\usepackage{mathtools}
\usepackage{braket}
\usepackage{float}
\usepackage{color}
\usepackage[usenames,dvipsnames]{xcolor}
\usepackage{hyperref}  
\usepackage[capitalize,nameinlink,noabbrev]{cleveref}
\usepackage{wrapfig}
\usepackage{tikz}
\usepackage[left=1in,right=1in,top=1in,bottom=1in]{geometry}

\hypersetup{colorlinks=true,urlcolor=[rgb]{0,0,0.5},citecolor=[rgb]{0.5,0,0},linkcolor=[rgb]{0,0,0.4}}

\DeclareUnicodeCharacter{00A0}{ }

\newcommand{\id}{\mathbbm{1}}

\newcommand{\ketscalar}[2]{\left\langle #1 \vert #2\right\rangle}

\newcommand{\R}{\mathbb{R}}

\newcommand{\C}{\mathbb{C}}
\newcommand{\Pp}{\mathbb{P}}
\newcommand{\E}{\mathbb{E}}

\newcommand{\Z}{\mathbb{Z}}

\newtheorem{theorem}{Theorem}

\newtheorem{lemma}{Lemma}
\newtheorem{proposition}{Proposition}

\theoremstyle{definition}
\newtheorem{remark}{Remark}

\usepackage[utf8]{inputenc}
\usepackage{todonotes}

\newcommand{\avg}[1]{\langle #1\rangle}

\newtheorem*{theorem*}{Theorem}

\renewcommand{\H}{\mathcal{H}} %Hilbert space
\newcommand{\dg}{\dagger}

\newcommand{\vol}{\mathrm{vol}}

\newcommand{\vep}{\varepsilon}

\newcommand{\h}[1]{\boldsymbol{#1}}
\newcommand{\U}{\h{U}} %
\newcommand{\V}{\h{V}} % 
\newcommand{\I}{\h{I}} % identity channel

\newcommand{\UU}{\mathcal{U}} %Unitary channels
\def\ketbra#1{|{#1}\rangle\!\langle{#1}|}

\def\Pr{\mathbb{P}}

\def\ep{\varepsilon}

\def\tr{{\rm tr}}

\newcommand{\trec}{t_{\mathrm{rec}}}

\renewcommand{\d}{\mathrm{D}}

\newcommand{\texit}{t_{\mathrm{exit}} }

\newcommand{\dtr}{D_{\mathrm{tr}}}
\newcommand{\dbures}{D_{\mathrm{B}}}

\newcommand{\npack}{{\cal N}_{\rm pack}}
\newcommand{\ncov}{{\cal N}_{\rm cov}}
\newcommand{\effsupp}{d_{\mathrm{supp}} }

\newcommand{\varsecond}{\Delta(H^{2})}
\newcommand{\varfourth}{\Delta(H^{4})}

\usepackage{graphicx}% Include figure files
\usepackage{dcolumn}% Align table columns on decimal point
\usepackage{bm}% bold math
\usepackage{hyperref}% add hypertext capabilities
\usepackage[capitalize,nameinlink,noabbrev]{cleveref}

\begin{document}

\preprint{APS/123-QED}

\title{Tight bounds on recurrence time in closed quantum systems}
%\title{Tight bounds on recurrence time in unitary quantum dynamics}% Force line breaks with \\

\author{Marcin Kotowski}
\email{\!mkotowski@cft.edu.pl}
\author{Michał Oszmaniec}
\email{\!oszmaniec@cft.edu.pl}
\affiliation{Center for Quantum Enabled-Computing, Center for Theoretical Physics of the
Polish Academy of Sciences, Al. Lotnik´ow 32/46, 02-668 Warsaw, Poland}
\date{}                     %% if you don't need date to appear

% \author{Ann Author}
%  \altaffiliation[Also at ]{Physics Department, XYZ University.}%Lines break automatically or can be forced with \\
%  \email{Second.Author@institution.edu}
% \affiliation{%
%  Authors' institution and/or address\\
%  This line break forced with \textbackslash\textbackslash
% }%

\date{\today}% It is always \today, today,

\begin{abstract}
The evolution of an isolated quantum system inevitably exhibits recurrence: the state returns to the vicinity of its initial condition after finite time. Despite its fundamental nature, a rigorous quantitative understanding of recurrence has been lacking. We establish upper bounds on the recurrence time, $t_{\mathrm{rec}} \lesssim \texit(\epsilon)(1/\epsilon)^d$, where $d$ is the Hilbert-space dimension, $\epsilon$ the neighborhood size, and $\texit(\epsilon)$ the escape time from this neighborhood. For pure states evolving under a Hamiltonian $H$, estimating $\texit$ is equivalent to an inverse quantum speed limit problem: finding upper bounds on the time a time-evolved state $\psi_t$ needs to depart from the $\epsilon$-vicinity of the initial state $\psi_0$. We provide a partial solution, showing that under mild assumptions $\texit(\epsilon) \approx \epsilon /\sqrt{ \Delta(H^2)}$, with $\Delta(H^2)$ the Hamiltonian variance in $\psi_0$. We show that our upper bound on $t_{\mathrm{rec}}$ is generically saturated for random Hamiltonians. Finally, we analyze the impact of coherence of the initial state in the eigenbasis of $H$ on recurrence behavior.

%The time evolution of a closed  quantum system exhibits recurrence - time-evolved state of the system will eventually return to the vicinity of the initial condition. The phenomenon of recurrence so far lacked precise quantitative understanding.  In this work we prove upper bounds on the recurrence time, $t_{\text{rec}}\lesssim \texit(\ep) (\frac{1}{\ep})^d $, where $d$ is the dimension of the system, $\ep$ is the size of the neighborhood, and $\texit(\ep)$ is the time the system needs to leave the $\ep$-neighborhood of the initial state.  Assessing the escape time for pure states evolving under a Hamiltonian $H$ requires solving the \emph{inverse quantum speed limit} problem - i.e. bounding the time the state $\psi_t$ needs to depart from the $\ep$-neighborhood of the initial state $\psi$. We give partial solution to this problem, showing that (under mild technical conditions) $t^{\text{state}}_{\text{esc}}(\ep)\approx \ep/\Delta_{\psi}(H^2)$, , where $\Delta_{\psi}(H^2)$ is the variance of the Hamiltonian in  $\psi_0$. We  construct examples of random Hamiltonians for which our bound on $t_{rec}$ are saturated.
%Finally, investigate the role that coherence of the initial state (with respect to the eigen-basis of $H$)  plays on recurrence.

%A state  of an isolated quantum system evolving according to time-independent hamiltonian $H$ will eventually return to the neighborhood of the initial state. This well-known phenomenon of recurrence however lacks precise quantitative  understanding. 

\end{abstract}

\maketitle

%\tableofcontents
%\onecolumngrid

\section{Introduction}

A state of a closed quantum system evolving under unitary evolution induced by a Hamiltonian operator $H$ in a finite-dimensional Hilbert space will eventually return arbitrarily close to its initial configuration after a sufficiently long time. This follows from the Poincaré recurrence theorem~\cite{Poincare1890,CornfeldFominSinai1982}, which applies to a broad class of dynamical systems, including closed quantum systems~\cite{ReccTheorem1950, schulman}. In generic cases, the recurrence (or Poincaré) time is expected to scale exponentially with the Hilbert-space dimension $d$, $t_{\mathrm{rec}} \sim e^{d}$. Previous works have provided physical arguments for this scaling (cf.~\cite{peres}) and estimated recurrence times under additional assumptions on the initial states or Hamiltonians~\cite{AverageRec,alvaro-recurrence,experimental1,Kac1943}. However, these approaches either (i) lack mathematical rigor, as they do not prove specific scaling relations for $t_{\mathrm{rec}}$; (ii) apply  to restricted classes of states or random Hamiltonians; or (iii) address recurrence statistics averaged over long times, rather than the \emph{first} recurrence event---which is of primary physical significance given the astronomical time scales involved. A complementary, model-specific line of work concerns revivals and collapse--revival phenomena in concrete systems, such as the Jaynes--Cummings model and interacting bosonic systems including confined Bose gases and two-site Bose--Hubbard dynamics~\cite{JaynesCummings1963,Pitaevskii1997,VekslerFishman2015}. These effects may be viewed as particularly regular and physically accessible manifestations of quantum recurrence, but their analysis relies on special features of the underlying models and therefore does not by itself provide general rigorous bounds on the first recurrence time.

This lack of rigorous quantitative understanding of recurrence time is unsatisfactory due to ubiquitous nature of recurrence (it was experimentally demonstrated in a suitably crafted many-body setting in \cite{expRECC}) and the role it plays in various areas, including the foundations of quantum statistical mechanics \cite{StatMech2016}, the phenomenology of circuit complexity in chaotic many-body dynamics \cite{BrownSuskind2017,complexityREVIEW,OKHH2024}  and its (conjectured) relation to quantum description of black holes and wormholes within AdS-CFT correspondence \cite{BHexp2020}. The purpose of this work is to clarify and advance precise understanding of recurrence phenomena in quantum mechanics.

\begin{figure}
    \centering
\includegraphics[scale=0.17]{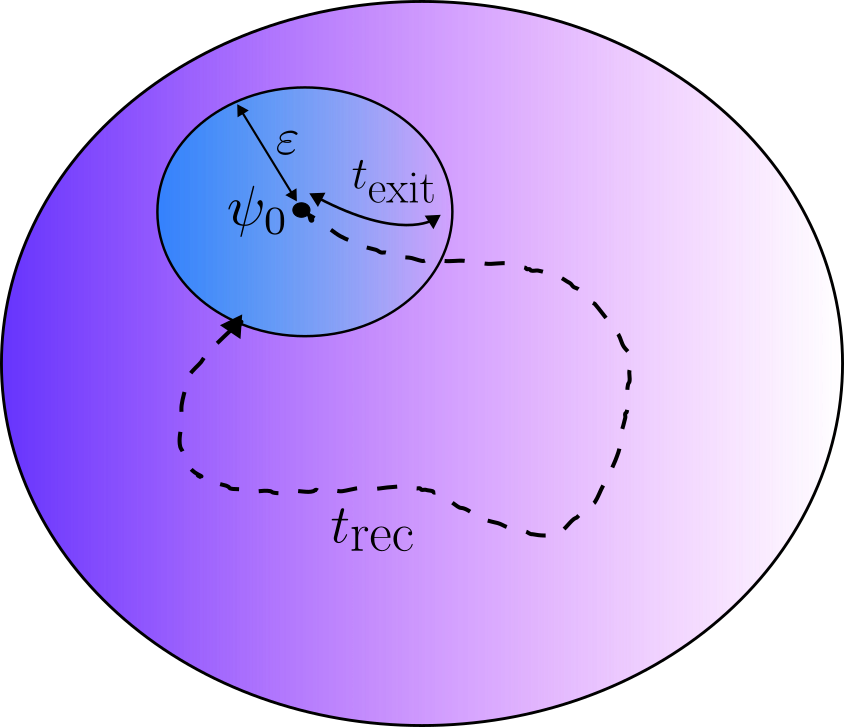}
    \caption{The Hamiltonian evolution of a closed quantum system started at an initial state $\psi_0$. After time $\texit$ the system leaves the $\vep$-ball around the initial point and subsequently after time $\trec > \texit$ the system has undergone recurrence.  }
    \label{fig:enter-label}
\end{figure}

 We first prove unconditional and deterministic upper bounds on the recurrence time, $t_{\text{rec}}\lesssim \texit(\ep) (\frac{1}{\ep})^d $, where $d$ is the dimension of the system, $\ep$ is the size of the neighborhood (measured in trace distance), and $\texit(\ep)$ is the time the system needs to leave the $\ep$-neighborhood of the initial pure state. 
Assessing the escape time requires solving the \emph{inverse quantum speed limit} problem - i.e. finding an upper bound on the time  system needs to depart from the $\ep$-neighborhood of the initial state $\psi_0$ (in the standard quantum speed limit \cite{MandelstamTamm1945} one typically is interested in finding a \emph{lower} bound on this time). We give partial solution to this problem, showing that under mild assumptions  $\texit(\ep)\approx \ep/\sqrt{\Delta(H^2)}$, where $\Delta(H^2)=\langle (H-\langle H\rangle)^2 \rangle $ is the Hamiltonian variance computed on the initial state. 

We emphasize that the recurrence time $\trec(\vep)$ has to be defined as larger than $\texit(\vep)$ -- this condition ensures that the evolution first leaves the neighborhood of $\psi_0$ and only later returns to it. Without imposing this we may have a situation in which evolution $\psi_t$ forever stays in the neighborhood of $\psi_0$ and no meaningful recurrence can be defined. Despite its simplicity this consideration has been absent from the study of recurrence. 

% \textcolor{red}{remove this paragraph} Our upper bounds on recurrence time are optimal in both $\epsilon$ and $d$- we construct examples of random Hamiltonians for which they are  are saturated. Additionally, we shed light on the role that coherence of the initial state $\ket{\psi}=\sum_k a_k \ket{k}$ (with respect to the eigen-basis of $H$). Concretely, we provide bounds on recurrence time involving  $\mathrm{supp}(\psi,\epsilon)$ - effective support \cite{support2019} of $\psi$, which counts the number of directions $i$ populated by $1-\epsilon$ fraction of probability distribution $p_k=|a_k|^2$.    Crucially, we show that for certain initial states and Hamiltonians the recurrence can take place in time exponentially larger than $(\frac{1}{\ep})^{d_{\text{eff}}}$ , where $d_{\text{eff}}^{-1}=\sum_{k} |\alpha_k|^4$ is the so-called effective dimension of $\psi$ \cite{EQ1,EQ2,EQ3}. A recent paper \cite{alvaro-recurrence}  proved \emph{lower bounds} on the average recurrence time $\langle t_{rec} \rangle \gtrsim e^{d_{\text{eff}}}$, and our result demonstrate that the bound form \cite{alvaro-recurrence} does not always capture the actual timescales of recurrence.  
 
 Our results can be strengthened if the state $\psi_0$ is concentrated on a small number of eigenstates of $H$ and generalize to Hamiltonian evolutions in the space of unitary channels equipped with the diamond norm. We can also handle subsequent recurrences, not only the first one. 
 Finally,  we  adapt our techniques cases where Hamiltonian $H$ describes free evolutions of systems of bosons or fermions \cite{Ventuli15,alvaro-recurrence,ReturnProbFree18,experimental1}, for which case recurrence times can be \emph{much shorter} as they are controlled by the dimension of the, single particle Hilbert space rather than dimension of the total multi-particle Hilbert space. 

 From the technical perspective our proofs are elementary --  they rely only on basic geometric properties of spaces in which quantum evolution takes place (the set of pure quantum states or unitary channels equipped with trace distance and diamond norm, respectively) and the fact that the unitary dynamics preserves their geometry.

\section{Notation and setup}

In this work we will be concerned with unitary evolutions in finite-dimensional Hilbert space $\H=\C^d$. These evolutions are solutions of the Schroedinger equation $i\frac{d}{dt}\ket{\psi_t}=H\ket{\psi_t}$, where $H$ is the Hamiltonian of the system.  Let  $\lbrace\ket{k}\rbrace_{k=1}^d$ and $\lbrace\lambda_k\rbrace_{k=1}^d$ denote the eigenstates and  the corresponding eigenvalues of $H$, respectively (we allow for arbitrary degeneracy in the spectrum of $H$).   For an initial state $\psi_0 = \ketbra{\psi_0}$ we have $\psi_t = e^{-iHt}\psi_0 e^{iHt}$. If we write  $\ket{\psi_0}$ in the eigenbasis of $H$ we get $ \ket{\psi_t} = e^{-it H} \ket{\psi_0}= \sum_{k=1}^{d}a_k e^{-i\lambda_k t}\ket{k}$, with $\sum_{k=1}^d |a_i|^2=1$. For an operator $A$ we will denote the expected value  $\tr(\psi_0 A)$   simply as $\avg{A}$. We will denote the variance $\avg{(H - \avg{H})^2}$ as $\varsecond$ and the fourth central moment $\avg{(H - \avg{H})^4}$ as $\varfourth$ .

Consider a metric space $(X,d)$ and suppose we have a time evolution $x_t$ starting from an initial point $x_0$. We define the {\bf exit time} $\texit(\vep)$ on scale $\vep$ as the first time when the time evolution leaves the $\vep$-neighborhood of the initial point:
\begin{equation}\label{eq:def-exit-time}
    \texit(\vep) = \inf\{ t>0 \ \vert \ d(x_0,x_t) \geq \vep \}\ .
\end{equation}
Accordingly, define the {\bf recurrence time} $\trec(\vep)$ on scale $\vep$ to be the first time after $\texit(\vep)$ when the system return to the $\vep$-neighborhood of the initial point:
\begin{equation}\label{eq:def-rec-time}
    \trec(\vep) = \inf\{ t > \texit(\vep) \ \vert \ d(x_0,x_t) \leq \vep \} \ .
\end{equation}
One can define subsequent recurrence times $t^{(k)}_{rec}(\vep)$ in a straightforward way -- the $k+1$'st recurrence is the first time of coming back to within distance $\vep$ of $x_0$ after: i) undergoing the $k$'th recurrence, and then: ii) leaving the neighborhood of $x_0$. We will mostly consider the case where $(X,d)$ is a subset of pure quantum states equipped with the trace distance $\dtr(\rho,\sigma)=\frac{1}{2}\|\rho-\sigma\|_1$  (which for pure quantum states reduces to $\dtr(\psi,\phi)=\sqrt{1-|\langle \psi|{\phi}\rangle|^2}$ ) and the time evolution is given by $\psi_t = e^{-iHt}\psi_0 e^{iHt}$. We can also consider the set of unitary channels equipped with the diamond distance, see Supplementary Material. 

%We emphasize that the recurrence time $\trec(\vep)$ has to be defined as larger than $\texit(\vep)$ -- this condition ensures that the evolution first leaves the neighborhood of $x_0$ and only later returns to it (without imposing this  we may have a situation in which evolution $x_t$  forever stays in the neighborhood of $x_0$ and  no meaningful recurrence can be defined) \mo{todo:maybe mention subsequen recurrencies}. 

\section{Results}

Our main result provides a general upper bound on the first recurrence time for any Hamiltonian and initial state. 
%An analogous result for unitary recurrence also holds, see \cref{th:recurrence-unitary} in the Supplemental Material.

\begin{theorem}[Recurrence for states]\label{th:recurrence-main-prl}
    For any $H$ and any initial state $\psi_0$ the first recurrence time satisfies:
    \begin{equation}\label{eq:recurrence-states-abstract}
       % \trec(\vep)  \leq  \texit(\vep) \cdot \left(\frac{4\pi}{\vep}\right)^{d-1}\ .
        \trec(\vep)  \leq  \texit(\vep) \cdot \left(\frac{2\pi}{\vep}+1\right)^{d-1}\ .
    \end{equation}
    Furthermore, the $k$'th recurrence time satisfies $\trec^{(k)}(\vep)\leq k\cdot \texit(2\vep) \left(\frac{4\pi}{\vep}+1\right)^{d-1}$.
\end{theorem}
To be useful, the above inequality requires an upper bound on the exit time $\texit(\vep)$. The following result provides such a bound, provided that $\vep$ is small enough. 

\begin{theorem}[Inverse quantum speed limit]\label{th:exit-time-bound-prl}
Assume that $\varsecond > 0$ (equivalently, $\ket{\psi_0}$ is not an eigenstate of $H$). Let $\vep_{\ast} = \frac{\varsecond}{\varsecond + \sqrt{\varfourth}}$. 
%        \vep_{\ast} = \frac{\avg{\Delta H^2}}{\avg{\Delta H^2} + \sqrt{\avg{(H-\avg{H})^4}}}  
    Then for any $\vep < \vep_{\ast}$ the exit time $\texit(\vep)$ is finite and satisfies
    % \begin{equation}
    %     \texit(\epsilon) \leq \frac{\vep}{\left( \avg{\Delta H^2} - \left( \avg{\Delta H^2} + \sqrt{\avg{(H-\avg{H})^4}} \right) \vep \right)^{1/2}} 
    % \end{equation}
        \begin{equation}\label{eq:exit-time-bound-prl}
        \texit(\epsilon) \leq \frac{\vep}{\sqrt{\varsecond}}\cdot\left( 1 -  \frac{\vep}{\vep_{\ast}} \right)^{-1/2} \ .
    \end{equation}
    %Moreover, the distance from $\psi_0$ is monotonically increasing until time $\texit(\vep_{\ast})$.
\end{theorem}

The above result depends on the Hamiltonian and the initial state via the variance $\varsecond$ and the fourth central moment $\varfourth$, quantities that together define curvature of time evolution $\psi_t$ in the space of pure quantum states \cite{Laba2017,Browne2024}. Importantly, the  bound from Eq. \eqref{eq:exit-time-bound-prl} is asymptotically optimal, because applying the well-known Mandelstam-Tamm quantum speed limit \cite{MandelstamTamm1945} yields a \emph{lower bound}: $\frac{\arcsin{\vep}}{\sqrt{\varsecond}} \leq \texit(\vep)$ which has a matching behaviour for small $\vep$. We discuss the general issue of finiteness of $\texit$ in the Supplemental Material.

To obtain a concrete bound on recurrence time, we can use the above theorems and put $\vep < \frac{\vep_{\ast}}{2}$. The exit time is then bounded explicitly:
%and we have the inequality:
    \begin{equation}\label{eq:recurrence-states-concrete}
        %\trec(\vep) \leq \frac{\vep}{\sqrt{2\varsecond}} \cdot \left(\frac{4\pi}{\vep}\right)^{d-1} \ .
        \trec(\vep) \leq \frac{\vep}{\sqrt{2\varsecond}} \cdot \left(\frac{2\pi}{\vep}+1\right)^{d-1} \ .
    \end{equation}

%\mo{comment on optimality of Theorem 2, geometry context: Hot start Dan Browne \cite{Browne2024}, geometry, torsion \cite{Laba2017}, geometric inference \cite{Brody1996}}

Up until now the question of upper bounding the first recurrence time for general Hamiltonians was mainly studied by Peres \cite{peres}, upon whose results we improve in several ways. The argument from \cite{peres} is only heuristic and applies only to sufficiently generic or random Hamiltonians -- in contrast, our results are fully rigorous and valid deterministically for {\it all} Hamiltonians. We can also handle arbitrary initial states, which enter the bounds via the exit time $\texit(\vep)$, while \cite{peres} only considers uniform superpositions of eigenstates and does not consider the issue of exit time. Our method is also general enough to encompass other settings (unitaries, highly symmetric Hamiltonians).  We note that the main idea of \cite{peres} of connecting recurrence times to certain geometric properties of the $d$-dimensional lattice can be made rigorous using Minkowski's theorem from convex geometry and has inspired our proof of lower bound in \cref{th:lower-bound-random-prl}. Finally, it can also be shown that the upper bound for subsequent recurrences from \cref{th:recurrence-main-prl}  is the best possible -- we discuss its tightness and possible strengthening in \cref{rm:torus-slow} in Appendix.

In our last result we show that the upper bound from \cref{th:recurrence-main-prl} is generically tight in the sense that for a simple model of random Hamiltonian a randomly chosen $H$ will saturate this bound with high probability. 
%\mo{(suggestion: remove in neccesary) Therefore, \cref{th:recurrence-main-prl} provides, up to constants, the best possible general upper bound on recurrence time.}

% SIMPLIFIED VERSION
\begin{theorem}\label{th:lower-bound-random-prl}
    Let $H$ be a random diagonal Hamiltonian with eigenvalues drawn independently and uniformly from $[-1,1]$. Let $\psi_0$ be the uniform superposition of the eigenstates of $H$. Then for $\vep$ small enough with probability at least $1-\exp(-\Omega(d))$ the recurrence time for the evolution started at $\psi_0$ is at least:
    \begin{equation}\label{eq:lower-bound-prl}
        \trec(\vep) \gtrsim  \left( \frac{1}{\vep}\right)^d\ .
    \end{equation}
\end{theorem}
The proof of the above theorem is technical and we present it in the Supplemental Material. 

% FORMAL VERSION - MOVE THIS TO SUPPLEMENT
% \begin{theorem}\label{th:lower-bound-random}
%     Let $H$ be a random Hamiltonian with eigenvalues $\lambda_k, k = 1,\dots, d$ as above. Let $\ket{\psi_0}$ be given by \eqref{eq:psi0-generic}. Assume $\vep <\frac{1}{10000}$ and $\vep > (40\pi)^{1/3}(1/2)^{d/3}$. Then with probability at least $1 - (3/4)^d - 5\exp(-\frac{d}{400})$ the recurrence time for the evolution started at $\ket{\psi_0}$ is at least:
%     \begin{equation}
%         \trec(\vep) >  20\left(  \frac{1}{10000}\cdot\frac{1}{\vep}\right)^d
%     \end{equation}
% \end{theorem}

\section{Proofs}

%and the covering number counts the minimal number of balls of size $\vep$ needed to cover the whole space
{\bf Sketch of proof of \cref{th:recurrence-main-prl}}. Let $X$ be a metric space equipped with a distance $d$ invariant with respect to the time evolution. The packing number $\npack(X,\vep)$ quantifies the maximal number of points in $X$ that are at least $\vep$-apart. See the appendix for detailed definitions. The core ingredient of the proof is the following bound, formally proved in \cref{prop:t+N} in the Appendix:
%The core of the proof is the following simple observation. For conceptual clarity we state it for abstract metric spaces, but one should think of $(X,d)$ as the orbit of the initial state under the Hamiltonian evolution, equipped with the trace distance, and the one-parameter group $\varphi_t$ as the time evolution $e^{iHt}$.
\begin{equation}
    \trec(\vep) \leq \texit(\vep) \cdot \npack(X,\vep).
\end{equation}
The idea behind the bound is to follow the time evolution $x_t$ started at the initial point $x_0$ and consider balls of radius $\vep$ around a suitable discretization $\lbrace x_{t_i}\rbrace$ of the trajectory. After sufficient time, the number of such balls will exceed the packing number $\npack(X,\vep)$, so some two points will end up in the same ball. Because we assume the distance is invariant, this implies that some point will end up near the origin and thus give recurrence (this step shares similarities  to the proof of the classic simultaneous Dirichlet approximation theorem from number theory, cf. \cite{Cassels1957})  . This logic is illustrated in \cref{fig:fig2}. A similar reasoning can be used to obtain bounds for the $k$'th recurrence time and is presented in the Supplemental Material. %To obtain the proof of \cref{th:recurrence-main-prl}, we take $(X,d)$ to be a subset of states reachable from the initial state, formally defined in \eqref{eq:torus-def}, with trace distance and $\varphi_t$ to be the time evolution $e^{iHt}$, which preserves the trace distance. To obtain \eqref{eq:recurrence-states-abstract}, we combine the above proposition with a bound on the packing number, which is proved in \cref{lm:packing-T} in the appendix. To obtain the concrete bound \eqref{eq:recurrence-states-concrete} we additionally use the exit time bound provided by \cref{th:exit-time-bound-prl} together with the condition $\vep < \frac{1}{2}\frac{\varsecond}{\varsecond + \sqrt{\varfourth}}$.
To obtain the proof of \cref{th:recurrence-main-prl}, we take $(X,d)$ to be (the closure of) the subset of states reachable from the initial state, formally defined in \eqref{eq:torus-def}, with trace distance and combine the above proposition with a bound on the packing number $\npack(\vep)$, which is proved in \cref{lm:packing-T} in the appendix. 
%To obtain the concrete bound \eqref{eq:recurrence-states-concrete} we additionally use the exit time bound provided by \cref{th:exit-time-bound-prl} together with the condition $\vep < \frac{1}{2}\frac{\varsecond}{\varsecond + \sqrt{\varfourth}}$.
% \textcolor{red}{abstract version only}
% \begin{proposition}\label{prop:t+N}
%  Consider either of the two actions of the one-parameter group $G=\{e^{iHt}\}_{t \in \R}$ on a metric space $(X,d)$, denoted by $\varphi_t$:
% \begin{enumerate}
%     \item $X = \UU_U$ equipped with the diamond distance, $x_0= I$ and $\varphi_t(\U) = \mathrm{Ad}_{e^{iHt}}\circ \U$ 
%     \item $X = \mathbb{T}_{\psi_0}$ equipped with the trace distance, $x_0 = \ketbra{\psi_0}$ and $\varphi_t(\ketbra{\psi}) = e^{iHt}\ketbra{\psi}e^{-iHt}$
% \end{enumerate}
% Then the recurrence time satisfies:
% \begin{equation}
%     \trec(\vep) \leq \texit(\vep) \cdot \npack(X,\vep)
% \end{equation}
% \end{proposition}

{\bf Proof of \cref{th:exit-time-bound-prl}}.   Let $F(t) := \tr( \psi_0 \psi_t)$ denote the fidelity between $\psi_t$ and the initial state $\psi_0$. The basic intuition is that for small $t$ we have $F(t) \approx 1- t^2 \Delta(H^2)$, so that to leave $\epsilon$ neighborhood in trace distance one needs time $\epsilon/\Delta(H^2)$  In what follows we formally prove that the same scaling of of $t_{esc}(\epsilon)$ holds also for finite $\epsilon$ (under the assumption of the theorem). 
%Let:
%\begin{equation}
%    X_H := -[H,[H,\psi_0]] =   2H\psi_0 H - H^2\psi_0 - \psi_0 H^2\ .
%\end{equation} 
Using $\frac{d}{dt}\psi_t = i[H, \psi_t]$ and simple algebra we get:
\begin{equation}\label{eq:f2ndt}
  F''(t) =  \tr(X_H  \psi_t)\ , X_H := -[H,[H,\psi_0]]\ .
\end{equation}
%equal to $F(t) := \tr( \psi_0 \psi_t)$. Since $\frac{d}{dt}\psi_t = i[H, \psi_t]$, we have:
%\begin{equation}
%    F'(t) = \tr(\psi_0 \frac{d}{dt}\psi_t) = i\tr(\psi_0 [H, \psi_t])
%\end{equation}
%and likewise:
%\begin{equation}
%   F''(t) =  -\tr(\psi_0 [H, [H, \psi_t]])
%\end{equation}
%Using the cyclic property of the trace we can move the commutators over to $\psi_0$ and write:
%\begin{equation}
%    F''(t) =  -\tr([H,[H,\psi_0]]  \psi_t)
%\end{equation}
%Let:
%\begin{equation}
%    X_H := [H,[H,\psi_0]] = H^2\psi_0 - 2H\psi_0 H + \psi_0 H^2
%\end{equation} 
%so that:
%\begin{equation}\label{eq:f2ndt}
%    F''(t) = -\tr(X_H \psi_t)
%\end{equation}.
%We will need a bound on the operator norm of $X_H$, which is supplied by the following lemma, proved in \mo{Supplemental Material}.
%We are now ready to prove \cref{th:exit-time-bound-prl}. Using \eqref{eq:f2ndt}, 
%where $ X_H := -[H,[H,\psi_0]]$.  
A quick computation shows  $\tr(X_H \psi_0)= -2\varsecond$ and therefore $F''(t) = -2\varsecond + \tr(X_H (\psi_t-\psi_0))$. Using the basic inequality:
% \begin{equation}
%     \abs{\tr(X_H (\psi_t-\psi_0))} \leq \norm{X_H}_{\infty}\norm{\psi_t - \psi_0}_1 = 2 \norm{X_H}_{\infty} \dtr(\psi_0, \psi_t)
% \end{equation}
$ \abs{\tr(X_H (\psi_t-\psi_0))} \leq  2 \norm{X_H}_{\infty} \dtr(\psi_0, \psi_t)
$ we get 
\begin{equation}\label{eq:2nd-derivative1}
    F''(t) \leq -2\varsecond +2 \norm{X_H}_{\infty} \dtr(\psi_0, \psi_t)\ .
\end{equation}
% \begin{equation}\label{eq:2nd-derivative2}
%     F''(t) \leq -2\avg{\Delta H^2} +2 \left( \avg{\Delta H^2} + \sqrt{\avg{(H - \avg{H})^4}} \right) \dtr(\psi_0, \psi_t)
% \end{equation}
From Lemma \ref{lm:norm-XH-prl} given in the appendix we get $\norm{X_H}_{\infty}\leq \varsecond + \sqrt{\varfourth}$.  Additionally,  before the exit time we have  $\dtr(\psi_0, \psi_t) \leq \vep$ and therefore for $t\in[0,\texit(\vep)]$  we have  % Using \cref{lm:norm-XH-prl}, we can bound \eqref{eq:2nd-derivative1} on the interval $[0, \texit(\vep)]$ as:
\begin{equation}\label{eq:fbisINEQ}
    F''(t) \leq -2\varsecond +2 \left( \varsecond + \sqrt{\varfourth} \right) \vep
\end{equation}
For simplicity let us write: 
\begin{equation}\label{eq:AB}
    A := \varsecond, \ B := \varsecond+ \sqrt{\varfourth}.
    \end{equation} 
    Integrating both sides of Eq. \eqref{eq:fbisINEQ} twice (each time from $0$ to $t\in[0,\texit(\vep)])$, and using $F'(0)=0, F(0)=1$ yields:
    %\begin{equation}\label{eq:derivative}
% F'(t) \leq \left( -2A +2  B \vep \right) t
%\end{equation}
\begin{gather}
 F'(t) \leq \left( -2A +2  B \vep \right) t  \label{eq:derivative} \\
    F(t) \leq 1 - \left( A - B \vep \right) t^2 \label{eq:ft-fidelity1} 
\end{gather}
By definition at time $\texit(\vep)$ we have $\dtr(\psi_0, \psi_{\texit}) = \vep$. Recalling  that for pure states $\dtr(\psi,\phi) = \sqrt{1 - F(\psi,\phi)}$, we have $F(\texit(\vep)) = 1 - \vep^2$. By plugging this into \eqref{eq:ft-fidelity1} at $t=\texit(\vep)$ we obtain:
\begin{equation}
 %   1 - \vep^2 \leq 1 - \left( \avg{\Delta H^2} - \left(  \avg{\Delta H^2} + \sqrt{\avg{(H - \avg{H})^4}} \right) \vep \right) \texit^2
    1 - \vep^2 \leq 1 - \left( A - B \vep \right) \texit(\vep)^2
\end{equation}
The assumption on $\vep$ implies that the coefficient of $\texit^2$ above is positive, which directly implies Eq. \eqref{eq:exit-time-bound-prl}. 
\section{Further consequences and applications}
{\bf Effective support and effective dimension.} We  observe that \cref{th:recurrence-main-prl} can be strengthened if the initial state $\psi_0$ is mostly supported on a small number of eigenvalues of $H$. In that case, the dynamics, and thus the problem of recurrence, can be restricted to that small subspace at the cost of a small error. Inspired by \cite{support2019} we define $\delta$-effective support of $\psi$ to be the size of the smallest subset $S$ of eigenvectors of $H$ that carry at least $1-\delta$ fraction of the total probability:
\begin{equation}
    \effsupp(\delta) := \min_{S}\left\{\abs{S} \ \vert \ \sum_{i \in S}\abs{a_i}^2 \geq 1-\delta\right\}
\end{equation}
Let $\psi^S$ denote the state obtained from $\psi$ by retaining only coefficients in $S$ and normalizing. It is easy to check that for all times $t$ we  have $F(\psi_t,\psi_t^S)\geq 1-\delta$.  By using $\dtr(\psi,\phi)^2 = 1 -F(\psi,\phi)$ and applying the triangle inequality twice we end up with:
\begin{equation}\label{eq:reduceddistances}
    \dtr(\psi_0,\psi_t) \leq \dtr(\psi^S_0,\psi^S_t) + 2\sqrt{\delta} \ .
\end{equation}
If we now take $\delta = \vep^2/16$, returning $\vep/2$-close to the initial point in the reduced dynamics $\psi^S_t$ implies returning $\vep$-close in the full dynamics $\psi_t$. Also, using \eqref{eq:reduceddistances} we get that exiting $\vep$ neighborhood of $\psi_0$ by the full dynamics $\psi_t$ implies that the reduced dynamics $\psi_t^S$ escaped $\vep/2$ neighborhood of $\psi^S_0$, implying $\texit^S(\ep/2)\leq \texit(\vep)$. Thus, we can apply \cref{th:recurrence-main-prl} to the reduced dynamics on scale $\vep/2$ and obtain:
\begin{equation}\label{eq:recurrence-states-abstract-reduced}
    \trec(\vep)  \leq  \texit(\vep) \cdot \left(\frac{4\pi}{\vep}+1\right)^{\effsupp(\vep^2/16)-1}\ .
\end{equation}
which improves upon \eqref{eq:recurrence-states-abstract} if the effective support is much smaller than $d$. %The exponent can be further improved to $\effsupp(\vep/2)$ by constructing appropriate coverings, see \cref{rm:eff-support} in Appendix.

Another parameter that has been considered in connection with recurrence is the so-called effective dimension of $\psi$ \cite{EQ1,EQ2,EQ3} equal to $d_{\text{eff}}^{-1}=\sum_{k} |\alpha_k|^4$. A recent paper \cite{alvaro-recurrence}  proved \emph{lower bounds} on the average recurrence time $\langle t_{rec} \rangle \gtrsim e^{d_{\text{eff}}}$. We show that this parameter not necessarily controls the true behavior of recurrence. Namely, in the Supplemental Material we provide a version of \cref{th:lower-bound-random-prl} with initial states $\psi_0$ which have small effective dimension (e.g. $\sim \log d$, and in fact can be taken to be growing arbitrarily slowly with $d$) compared to its effective support ($\sim d$) and satisfies \eqref{eq:lower-bound-prl}. This shows that it's the effective support, rather than the effective dimension, that controls the lower bound on recurrence time. See \cref{th:lower-bound-random-eta} in the Supplemental Material for technical statement.

{\bf Symmetric Hamiltonians.} We obtain much stronger bounds on $\trec$ for free evolutions of multiparticle systems. Specifically, assume that an $n$ particle state $\psi_t$ evolves under a non-interacting Hamiltonian: 
\begin{equation}\label{eq:nonInteractingHamiltonian}
    H_n=H\otimes\overbrace{\mathbb{I}\otimes \ldots \otimes\mathbb{I}}^{n-1} + \mathbb{I}\otimes H \otimes\overbrace{\mathbb{I} \otimes\ldots \otimes\mathbb{I}}^{n-2}+\ldots\ ,
\end{equation}
where $H$ is a Hamiltonian on a single particle space $\mathbb{C}^d$. Then, the recurrence time of $\psi_t$ is bounded by 
\begin{equation}\label{eq:free_recurrence}
    \trec(\vep)  \leq  \texit(\vep) \cdot \left(\frac{4\pi n}{\vep}\right)^{d-1}\ ,
\end{equation}
which with increasing $n$ significantly improves over Eq. \eqref{eq:recurrence-states-abstract} as dimension of a single particle space $d$ is much smaller than that of $n$ qudit Hilbert space $(\mathbb{C}^d)^{\otimes n}$ ($d^n$), or its fermionic and bosonic subspaces (of dimensions $\binom{d}{n}$ and $\binom{n+d-1}{n}$ respectively). To obtain Eq. \eqref{eq:free_recurrence} we follow the sketch of the proof of \cref{th:recurrence-main-prl} and take $X$ as the space of unitary channels on $(\mathbb{C}^d)^{\otimes n}$ corresponding to unitaries of the form $U^{\otimes n}$, with $U=\sum_{k=1}^d \exp(i \varphi_k) \ketbra{k}$ equipped with the diamond norm (see Supplemental Material for details). A similar result can be derived also for evolutions generated by fermionic quadratic Hamiltonians acting on the fermionic Fock space.

\section{Conclusions and open problems}

We have presented rigorous quantitative bounds on the recurrence times in unitary dynamics of pure quantum states in finite-dimensional quantum systems that incorporate the, previously not investigated, issue of finding an upper bound on the exit time $\texit$ (inverse quantum speed limit). Our results apply deterministically to arbitrary initial states and Hamiltonians and can be further tailored to initial states with support on a small number of eigenstates as well as when the Hamiltonian describes evolution of non-interacting particles. 
Our work leaves a number of open problems for further work. It would be interesting to generalize our deterministic bounds to dynamics of mixed quantum states, expectation values of observables on evolved state and prove tight recurrence bound for physical dynamics induced by local non-commuting Hamiltonians. It also seems natural, especially in the context of dynamics after a quench \cite{quench1,quench2,quench3}, to extend the results regarding inverse quantum speed limits to mixed states, open quantum system dynamics and dynamics of expectation values. 
%computational aspects She-Yuen \cite{sheYuen23} - \textcolor{red}{In discussion}

\begin{acknowledgments}

We thank Alvaro Alhambra, Mischa Woods, Martin Kliesch, Maciej Lewenstein and Daiki Suruga for interesting discussions. The authors acknowledges support from National Science Center, Poland within the QuantERA III Programme (No 2023/05/Y/ST2/00140 acronym Tuquan). The C4QEC project is carried out within the IRAP of the Foundation for Polish Science co-financed by the European Union.
\end{acknowledgments}

%\newpage
%\onecolumngrid
%\appendix

%\section{Appendix: old text}

\nocite{*}
\bibliography{biblio}% Produces the bibliography via BibTeX.

\clearpage

\part*{End matter}\label{sec:appendix}

For a metric space $(X,d)$ we formally define packing and covering numbers as:
\begin{align} 
    \npack(X,\vep) &:= \max\big\{|S|: \forall_{x,y\in S}\ d(x,y)\geq \vep, ~S\subseteq X\big\} \nonumber \\ \nonumber
    \ncov(X,\vep) &:= \min\left\{|S|:X = \bigcup_{x\in S} B(x,\vep), ~S\subseteq X\right\} \nonumber.
\end{align}
It is a standard textbook fact (see e.g. \cite{vershynin}, Lemma 4.2.8) that covering and packing numbers are related as follows:
\begin{align}
\label{eq:pack-cov}
    \ncov(X,2\vep)\leq \npack(X,\vep) \leq \ncov(X,\vep)\,.
\end{align}

We now formally state and prove the core proposition used in the proof of \cref{th:recurrence-main-prl} for the first recurrence time. A version covering also the case of subsequent recurrences is shown in the Supplemental Material.

\begin{proposition}\label{prop:t+N}
Let $(X,d)$ be a metric space and let $\varphi_t$ be an action on $X$ of a one-parameter group that leaves the metric invariant, i.e. $d(\varphi_t(x), \varphi_t(y)) = d(x,y)$. Then the first recurrence time satisfies:
\begin{equation}\label{eq:first-recurrence}
    \trec(\vep) \leq \texit(\vep) \cdot \npack(X,\vep)
\end{equation}

%Then for any $k \geq 1$ the $k$'th recurrence time satisfies:
% \begin{equation}
%     \trec(\vep) \leq k\cdot \texit(2\vep) \cdot \npack(X,\vep/2)
% \end{equation} 

\end{proposition}

\begin{proof}

 Suppose first that $\varphi: X\to X$ is an isometry of $X$. We claim that for any $x_0 \in X$  there exists some $1 \leq k \leq  \npack(X,\vep) $ such that:
    \begin{equation}
        d(x_0, \varphi^k(x_0)) < \vep
    \end{equation}
    where $\varphi^k$ denotes the $k$-fold composition of $\varphi$.

 To prove this, let $N=\npack(X,\vep)$ and consider the sequence of points $x_k= \varphi^k(x_0) $ for $ k=0,\dots,N$. Since the number of points is $N+1 > \npack(X,\vep)$, by the definition of the packing number there exist some $i < j$ such that $d(x_i,x_j) < \vep$, which translates to:
    \begin{equation}
    d(\varphi^i(x_0), \varphi^j(x_0)) < \vep
    \end{equation}
    Since $\varphi$ is an isometry of a compact metric space, it is invertible and the inverse $\varphi^{-1}$ is also an isometry. By applying $\varphi^{-i}$ to both sides we obtain:
    \begin{equation}
     d(x_0, \varphi^{j-i}(x_0)) < \vep
    \end{equation}
    so the claim follows with $k=j-i$.

    Now, by assumption for any $t$ the map $\varphi_t$ is an isometry with respect to $d$. Let $\varphi := \varphi_{\texit}$. By applying the above claim to $\varphi$ we obtain that for some $1 \leq k \leq \npack(X,\vep)$ we have:
    \begin{equation}
        d(x_0 , \varphi_{k\texit}(x_0) ) < \vep
    \end{equation}
    so a recurrence has occurred at time $t = k\texit$, which is indeed a recurrence since $t \geq \texit$.

%     We use the same logic to obtain bounds for $k$'th recurrence for $k \geq 2$. Consider $N_k = k \npack(X,\vep/2)$ and let $\mathcal{N} \subseteq \{x_i\}_{i=0}^{N_k}$ be any maximal set of points which are pairwise $\geq\vep/2$ apart, so that $\abs{\mathcal{N}} \leq \npack(X,\vep/2)$. By the pigeonhole principle, there will be some  $k+1$ points $x_{i_{1}},\dots,x_{i_{k+1}}$, with $i_1 < \dots < i_{k+1}$  and with $x_{i_{j}} \in \mathcal{N}$ for some $j$, such that $d(x_{i_{j}}, x_{i_{j'}}) < \vep/2$ for all $j'=1,\dots,k+1$. It follows that $d(x_{i_{1}},x_{i_{j'}}) < \vep$, so by invariance $d(x_{0},x_{i_{j'}-i_{1}}) < \vep$ and we have $k$ distinct points within distance $\vep$ of $x_0$.

%     To show that these points are in fact recurrences we need to take $\varphi := \varphi_{t}$ with $t=\max_k \texit(x_k,x_0,\vep)$, where $\texit(x,y,\vep)$ is the time to leave the neighborhood of $y$ when started from $x$. Clearly we have the following bound:
% \begin{equation}
%     \texit(x_k, x_0, \vep) \leq \texit(x_k, x_k, 2\vep) = \texit(x_0, x_0, 2\vep) 
% \end{equation}
% where the first inequality holds since having left the ball of radius $2\vep$ around $x_k$ implies having left the ball of radius $\vep$ around $x_0$ and the second equality follows from the invariance of dynamics.
\end{proof}

\begin{figure}
    \centering
    \includegraphics[scale=0.25]{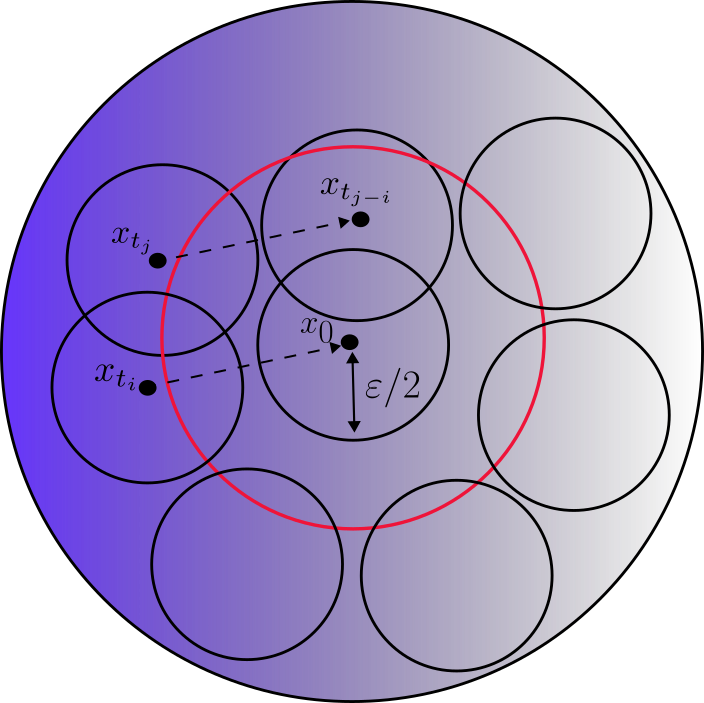}
    \caption{Main part of the proof of \cref{prop:t+N}. If the number of points $x_k= \varphi^k(x_0)$ is greater than the packing number $\npack(X,\vep)$, there exists a point $x_{j-i}$ within distance $\vep$ of $x_0$. The ball of radius $\vep$ around $x_0$ is shown in red. }
    \label{fig:fig2}
\end{figure}

%For an initial state $\psi_0$ written as in \eqref{eq:coefficients} let $\mathbb{T}_{\psi_0}$ be the subset of quantum states defined as:
For an initial state $\psi_0$  let $\mathbb{T}_{\psi_0}$ be the subset of quantum states defined as:
\begin{equation}\label{eq:torus-def}
%    \mathbb{T}_{\psi_0} := \{ \ketbra{\psi} \in \S(d) \ \vert \ \ket{\psi}= \sum_{k=1}^{d} a_k e^{i\varphi_k} \ket{k}, \varphi_k \in [0,2\pi) \}
    \mathbb{T}_{\psi_0} := \{ \ketbra{\psi} \ \vert \ \ket{\psi}= \sum_{k=1}^{d} a_k e^{i\varphi_k} \ket{k}, \varphi_k \in [0,2\pi) \}
\end{equation}
where $a_k = \langle k|{\psi_0}\rangle$ and $\ket{k}$ are the eigenstates of $H$. 

It will be convenient to also use the Bures distance between states, $\dbures(\psi,\phi) = \sqrt{2 - 2\abs{\langle \psi \vert \phi \rangle}}$, with the following easily proved inequalities:
\begin{equation}\label{eq:dtr-hs}
% \frac{1}{\sqrt{2}}  \inf_{\varphi}\norm{\ket{\psi} - e^{i\varphi}\ket{\phi}}_{2}\leq \dtr(\psi, \phi) \leq \inf_{\varphi}\norm{\ket{\psi} - e^{i\varphi}\ket{\phi}}_{2}
\frac{1}{\sqrt{2}}  \dbures(\psi,\phi)\leq \dtr(\psi, \phi) \leq  \dbures(\psi,\phi).
\end{equation}
%It is straightforward to express the Bures distance using the Euclidean norm between vectors: $\dbures(\psi,\phi) = \inf_{\varphi}\norm{\ket{\psi} - e^{i\varphi}\ket{\phi}}_{2}$, which will be useful later.

% It will be convenient to also use the (projectivized) Euclidean norm between vectors representing states, with the following easily proved inequalities:
% \begin{equation}\label{eq:dtr-hs}
% \frac{1}{\sqrt{2}}  \inf_{\varphi}\norm{\ket{\psi} - e^{i\varphi}\ket{\phi}}_{2}\leq \dtr(\psi, \phi) \leq \inf_{\varphi}\norm{\ket{\psi} - e^{i\varphi}\ket{\phi}}_{2}
% \end{equation}

\begin{lemma}\label{lm:packing-T}
    The packing number $\npack( \mathbb{T}_{\psi_0},\vep)$ with respect to the trace distance satisfies:
    \begin{equation}
%        \npack( \mathbb{T}_{\psi_0},\vep) \leq \left( \frac{4\pi}{\vep} \right)^{d-1}
       \npack( \mathbb{T}_{\psi_0},\vep) \leq \left( \frac{2\pi}{\vep}+1 \right)^{d-1}
    \end{equation}
\end{lemma}

\begin{proof}
     By \eqref{eq:pack-cov} it suffices to upper bound the covering number  $\ncov(\mathbb{T}_{\psi_0},\vep)$. We consider $\mathbb{T}_{\psi_0}$ equipped with the Bures metric $\dbures$.
% \begin{equation}\label{eq:d-inf-phi}
%     D( \ketbra{\psi}, \ketbra{\psi'} ) = \inf_{\varphi}\norm{ \ket{\psi} - e^{i\varphi}\ket{\psi'} }_2
% \end{equation}
By inequalities \eqref{eq:dtr-hs} any $\vep$-covering of $\mathbb{T}_{\psi_0}$ with respect to the Bures distance is an $\vep$-covering with respect to the trace distance, so it suffices to upper bound $\ncov(\mathbb{T}_{\psi_0},\vep)$ with respect to the Bures distance.
    
Let $Y_{\psi_0}$ be the set of vectors $\ket{\psi}$ of the form:    
\begin{equation}
    \ket{\psi} = \sum_{k=1}^{d} a_k e^{i\varphi_k} \ket{k}
\end{equation}
with $\varphi_1=0$, equipped with the Euclidean metric $\norm{ \ket{\psi} - \ket{\psi'}}_2$. Let $P: Y_{\psi_0}\to \mathbb{T}_{\psi_0}$ be the quotient map $\ket{\psi}\to\ketbra{\psi}$. We easily check that $P$ is distance-nonincreasing. For such maps it is evident that:
\begin{equation}\label{eq:ncov-P}
    \ncov(P(X),\vep) \leq \ncov(X,\vep)
\end{equation}
Since $P$ is surjective, to bound the covering number of $\mathbb{T}_{\psi_0}$ it thus suffices to provide a covering for $Y_{\psi_0}$ with respect to the Euclidean distance. Let $n=\lceil \frac{2\pi}{\vep}\rceil$ and let $N$ be the set of states with phases of the form $x=(1,x_{i_{2},2}, \dots, x_{i_{d},d})$, where each $i_k=2,\dots, n$ and $x_{j,k} = e^{i \frac{2\pi j}{n}}$. We claim that the set $N$ forms an $\vep$-covering of $Y_{\psi_0}$. To show this, first observe that for two vectors $\ket{\psi}, \ket{\psi'} \in Y_{\psi_0}$:
\begin{equation} \label{eq:l2-phases}
    \begin{split}
    &\norm{\ket{\psi} - \ket{\psi'}}^2_2  = 
    2\sum_{k=2}^{d}\abs{a_k}^2 (1-\cos(\varphi_k-\varphi'_k)) = \\
    &4\sum_{k=2}^{d}\abs{a_k}^2 \sin^2\left( \frac{\varphi_k-\varphi'_k}{2} \right) \leq \sum_{k=2}^{d}\abs{a_k}^2 (\varphi_k-\varphi'_k)^2
\end{split} 
\end{equation}
where we have used the assumption $\varphi_1 = \varphi'_1 =0$. By construction of the set $N$ for every $k\geq 2$ the point $\varphi_k$ is $\vep$-close on the unit circle to $\varphi'_k = x_{i_k,k}$ for some $i_k$.  Thus:
    \begin{equation}
%        \ncov( \mathbb{T}_{\psi_0}, \vep) \leq n^{d-1} =  \left\lceil \frac{2\pi}{\vep}\right\rceil ^{d-1} \leq \left( \frac{4\pi}{\vep} \right)^{d-1}
        \ncov( \mathbb{T}_{\psi_0}, \vep) \leq n^{d-1} =  \left\lceil \frac{2\pi}{\vep}\right\rceil ^{d-1} \leq \left( \frac{2\pi}{\vep} +1\right)^{d-1}
    \end{equation}
    which finishes the proof.
\end{proof}

\begin{remark}\label{rm:torus-slow}
Our bound on subsequent recurrence times from \cref{th:recurrence-main-prl} can be interpreted as saying that each recurrence takes time $\lesssim (1/ \vep)^d$ on average. It is natural to ask whether this can be strengthened to saying that each recurrence simply takes time bounded by $\lesssim (1/\vep)^d$ (without the average) The following counterexample, involving slow and fast timescales, shows that this cannot be done in general. Fix $\vep$ and consider a qutrit state $\ket{\psi_t} = a_1 \ket{1} + a_2 e^{-i\lambda_2 t} \ket{2} + a_3 e^{-i\lambda_3 t} \ket{3}$ with $\abs{a_1}^2=\frac{1}{2},\abs{a_2}^2=\frac{1}{2}- \delta, \abs{a_3}^2=\delta$ with $\delta \sim \vep^2$ and $\lambda_3 \gg \lambda_2$. By using the identity $\dtr(\psi_t, \psi_0)^2 = 1 - F(\psi_t,\psi_0)$ one can easily see that the evolution has two regimes. For times $t \lesssim \frac{\vep}{\lambda_2}$ the $a_2$ term is roughly constant and exits and recurrences happen on timescale $\frac{1}{\lambda_3}$. Once $t \sim \frac{\vep}{\lambda_2}$, the system will not undergo another recurrence for time $\sim \frac{1}{\lambda_2}$. Therefore, the time of the slow recurrence can be made arbitrarily larger than any function of $\vep$ and $\texit \sim \frac{1}{\lambda_3}$ by tuning the ratio $\frac{\lambda_3}{\lambda_2}$. This precludes a uniform bound on each recurrence that would depend only on $\vep$ and $\texit$, while it still conforms to the bound from \cref{th:recurrence-main-prl}. This example has effective support equal to $2$ and saturates the version of an upper bound on subsequent recurrences with effective support instead of dimension.
\end{remark}

\begin{lemma}\label{lm:norm-XH-prl}
    We have $\norm{X_H}_{\infty} \leq \varsecond + \sqrt{\varfourth}$.
\end{lemma}

\begin{proof}
    Recall that $X_H = H^2\psi_0 - 2H\psi_0 H + \psi_0 H^2$. Assume first that $H^2\ket{\psi_0} \neq 0$. We first compute the eigenvalues of $A = H^2\psi_0 + \psi_0 H^2$. It is clearly self-adjoint and of rank at most two, with image spanned by $\ket{\psi_0}$ and $H^2 \ket{\psi_0}$. The matrix of $A$ in this basis reads as (note that the matrix below is not self-adjoint since the basis need not be orthogonal):
    \begin{equation}
        \begin{pmatrix}
             \avg{H^2} & \avg{H^4} \\
             1 & \avg{H^2} \\
        \end{pmatrix}
    \end{equation}
    The eigenvalues of $A$ are thus $\lambda_{\pm} = \avg{H^2} \pm \sqrt{\avg{H^4}}$. This formula holds also if $H^2\ket{\psi_0}=0$ since then both sides are equal to zero.

    Since $2 H\psi_0 H \geq 0$ we can bound:
    \begin{equation}\label{eq:xh-comp}
        (\avg{H^2} - \sqrt{\avg{H^4}})I - 2H\psi_0 H \leq X_H \leq (\avg{H^2} + \sqrt{\avg{H^4}})I
    \end{equation}
    The operator $H\psi_0 H$ is of rank one and its nonzero eigenvalue is equal to $\avg{H^2}$. Altogether by \eqref{eq:xh-comp} this implies that:
    \begin{equation}
        (-\avg{H^2} - \sqrt{\avg{H^4}})I \leq X_H \leq (\avg{H^2} + \sqrt{\avg{H^4}})I
    \end{equation}
    which implies that $\norm{X_H}_{\infty} \leq \avg{H^2} + \sqrt{\avg{H^4}}$. Finally, we observe that $X_H$ is invariant with respect to shifting $H$ by a multiple of identity, so we can use this gauge freedom to obtain $\avg{\Delta H^2} + \sqrt{\avg{(H-\avg{H})^4}}$ on the right hand side, which finishes the proof. 
\end{proof}

\newpage
\onecolumngrid

\setcounter{theorem}{0}
\renewcommand{\thetheorem}{S\arabic{theorem}}
\renewcommand{\thelemma}{S\arabic{lemma}}
\renewcommand{\theproposition}{S\arabic{proposition}}
\renewcommand{\theassumption}{S\arabic{assumption}}

\setcounter{section}{0}

\part*{Supplemental Material}

\section{Subsequent recurrences}

We prove the claimed bound on the subsequent recurrence times after the first recurrence. The logic of the proof is the same as for the first recurrence -- the only difference is that we need a version of \cref{prop:t+N} for the $k$'th recurrence time. We claim that for any $k \geq 1$ the $k$'th recurrence time satisfies:
\begin{equation}\label{eq:subseq-packing}
    \trec^{(k)}(\vep) \leq k\cdot \texit(2\vep) \cdot \npack(X,\vep/2)
\end{equation} 
The bound from \cref{th:recurrence-main-prl} follows by inserting the packing number bound, as for the first recurrence.

The proof of \eqref{eq:subseq-packing} follows from a similar reasoning as the proof of \cref{prop:t+N}. Suppose that $\varphi: X\to X$ is an isometry of $X$. Fix $k\geq 1$ and let $N_k := k\cdot\npack(X,\vep/2)$. Fix $x_0 \in X$ and let $x_l := \varphi^l(x_0)$. We claim that there exist $k+1$ points $x_{i_{j}}$, with $0=x_{j_{0}} < i_1 < \dots < i_k \leq N_k$, such that $d(x_0,x_{i_{j}}) < \vep$. Let $\mathcal{N} \subseteq \{x_i\}_{i=0}^{N_k}$ be any maximal set of points which are pairwise $\geq\vep/2$ apart, so that $\abs{\mathcal{N}} \leq \npack(X,\vep/2)$. By the pigeonhole principle, there will be some  $k+1$ points $x_{l_{1}},\dots,x_{l_{k+1}}$, with $l_1 < \dots < l_{k+1}$  and with $x_{l_{j}} \in \mathcal{N}$ for some $j$, such that $d(x_{l_{j}}, x_{l_{j'}}) < \vep/2$ for all $j'=1,\dots,k+1$. From the triangle inequality it follows that $d(x_{l_{1}},x_{l_{j'}}) < \vep$, so by invariance $d(x_{0},x_{l_{j'}-l_{1}}) < \vep$, so it suffices to put $i_j := l_{j+1} - l_{1}$ for $j=0,\dots,k$.

% $T=\max\{\texit(x_{i_{j}},x_0,\vep)\}_{j=0}^{k-1}$
To obtain recurrence from above claim we apply it with $\varphi := \varphi_{T}$, where $T=\texit(2\vep)$. Let $\texit(x,y,\vep)$ be the time to leave the neighborhood of $y$ when started from $x$. Clearly we have the following bound:
 \begin{equation}
     \texit(x_k, x_0, \vep) \leq \texit(x_k, x_k, 2\vep) = \texit(x_0, x_0, 2\vep) 
 \end{equation}
where the first inequality holds since having left the ball of radius $2\vep$ around $x_k$ implies having left the ball of radius $\vep$ around $x_0$ and the second equality follows from the invariance of dynamics. The above inequality and the choice of $T$ imply that the time between $x_{i_{j}}$ and $x_{i_{j+1}}$ is at least $\texit(x_{i_{j}},x_0,\vep)$, which implies that each point $x_{i_{j}}$ is a new recurrence and thus finishes the proof.

\section{Recurrence for dynamics in the space of unitary channels}
%\section{Notations and tools for unitaries}

Every unitary operator $U$ on $\H$ defines a unitary channel $\U$ via conjugation $\U(\rho)=U\rho U^\dg$. Note that all unitary operators $\exp(i \varphi) U$, which differ from $U$ by a global phase, define the same unitary channel $\U$. We denote the set of quantum channels on $\H$ by $\UU(d)$. We choose to work with the following definitions of distance between unitary channels:
\begin{equation}\label{eq:distancesDEF}
    \d(\U,\V)= \min_{\varphi\in[0,2\pi)} \|U-\exp(i \varphi) V \| _\infty,
\end{equation}
where $\|\cdot\|_\infty$ is the operator norm. This distance describes the optimal statistical distinguishability of unitary channels, which follows from the relation between the operator norm and the diamond norm:
\begin{equation}\label{eq:equivDISTANCE}
\d(\U,\V) \leq \|\U-\V\|_\diamond \leq 2\, \d(\U,\V)\,,
\end{equation}
for a proof of this see, for instance \cite{OSH2020}.

Fix a Hamiltonian $H$. Let $\U_t$ denote the quantum channel corresponding to the time evolution unitary $U_t = e^{iHt}$. Note that $U_0 = I$. The exit and recurrence times in this setting are defined by \eqref{eq:def-exit-time} and \eqref{eq:def-rec-time}. We now restate the main upper bound from \cref{th:recurrence-main-prl} for the case of unitary channels.

 \begin{theorem}[Recurrence for unitaries]\label{th:recurrence-unitary}
    Consider the set of quantum channels equipped with the diamond norm. For any Hamiltonian $H$ with eigenvalues $\lambda_{min}=\lambda_1 \leq\dots\leq \lambda_d = \lambda_{max}$, we have :
    \begin{equation}
%        \trec(\vep) < \texit(\vep)  \left( \frac{8\pi}{\vep}\right)^{d-1} 
       \trec(\vep) < \texit(\vep)  \left( \frac{4\pi}{\vep}+1\right)^{d-1} 
    \end{equation}
    with $\texit(\vep) <  \frac{\pi \vep}{ \lambda_{max} - \lambda_{min}} $, so:
    \begin{equation}
%        \trec(\vep) <  \frac{\pi \vep}{ \lambda_{max} - \lambda_{min}}   \left( \frac{8\pi}{\vep}\right)^{d-1}  
        \trec(\vep) <  \frac{\pi \vep}{ \lambda_{max} - \lambda_{min}}   \left( \frac{4\pi}{\vep}+1\right)^{d-1}  
    \end{equation}
\end{theorem}

As in the case of states, the proof of \cref{th:recurrence-unitary} proceeds by using \cref{prop:t+N}. Thus, we need upper bounds on exit time and the packing number. We first bound the exit time, which is considerably easier than for states.

\begin{lemma}\label{lm:unitary-exit-time}
    The exit time $\texit$ for unitary channels equipped with the diamond distance is bounded as:
    \begin{equation}
        \texit(\vep)\leq \frac{\pi \vep}{\lambda_{max} - \lambda_{min}}
    \end{equation}
\end{lemma}    

\begin{proof}
Let us consider the Hamiltonian $H'$ equal to the original Hamiltonian $H$ shifted by a multiple of identity so that $\lambda'_{max} = -\lambda'_{min}$. Obviously $H$ and $H'$ generate the same dynamics. Recall the definition of the distance $\d$ from \eqref{eq:distancesDEF}. We have:
    \begin{equation}\label{eq:distance-unitary}
        \d(\U_t, \I) = \inf_{\varphi \in [0,2\pi)} \max_k \abs{e^{i\lambda'_k t} - e^{i\varphi}} =\inf_{\varphi \in [0,2\pi)} \max_k 2\abs{\sin\left( \frac{\lambda'_k t - \varphi}{2} \right)}
    \end{equation}
    
Let us define $\texit^{D} := \inf \{ t > 0 \ \vert \ \d(\U_t, \I) > \vep \}$.  
Obviously by \eqref{eq:equivDISTANCE} we have $\texit \leq \texit^{D}$. To upper bound the exit time $\texit$ it thus suffices to find the first time $t>0$ for which:
\begin{equation}
\inf_{\varphi \in [0,2\pi)} \max_k \abs{\sin\left( \frac{\lambda'_k t - \varphi}{2}  \right)} > \frac{\vep}{2}
\end{equation}

We claim that this time is equal to $t_0 = \frac{2x_0}{\lambda'_{max}}$, where $x_0$ is the smallest positive $x$ such that $\sin(x) = \frac{\vep}{2}$. Since $\lambda'_{max} = -\lambda'_{min}$, for any $\varphi$ at least one of $\abs{\sin\left( \frac{\lambda'_{max} t_0 - \varphi}{2}  \right)} = \abs{\sin\left( x_0 - \frac{\varphi}{2}  \right)}$ and $\abs{\sin\left( \frac{\lambda'_{min} t_0 - \varphi}{2}  \right)} = \abs{\sin\left( x_0 + \frac{\varphi}{2}  \right)}$ will be larger than $\sin(x_0)$. It is easy to see that $t_0$ is the smallest possible $t$ and that:
\begin{equation}
t_0 \leq   \frac{\pi}{2}\cdot\frac{\vep}{\lambda'_{max}}
\end{equation}
The lemma follows by shifting back to the original Hamiltonian $H$ so that $\lambda'_{max} = \frac{\lambda_{max} - \lambda_{min}}{2}$.
\end{proof}

We will now prove an upper bound on the packing number. Let $H=U \Lambda U^{\dagger}$ be the diagonalization of $H$ and let $\UU_U \subseteq \UU(d)$ denote the subset of unitary channels induced by unitary matrices of the form $UDU^{\dagger}$, where $D$ is diagonal. 

\begin{lemma}\label{lm:packing-UU}
    The packing numbers $\npack(\UU_U,\vep)$ with the diamond distance satisfies:
    \begin{equation}
     %   \npack(\UU_U,\vep) \leq \left( \frac{8\pi}{\vep} \right)^{d-1}
        \npack(\UU_U,\vep) \leq \left( \frac{4\pi}{\vep} + 1 \right)^{d-1}
    \end{equation}
\end{lemma}

\begin{proof}

    By \eqref{eq:pack-cov} it suffices to upper bound the covering number  $\ncov(\UU_U,\vep)$. Recall the definition of distance $\d$ from \eqref{eq:distancesDEF}. By inequalities \eqref{eq:equivDISTANCE}, any $\vep/2$-covering of $\UU_U$ with respect to $\d$ is an $\vep$-covering of $\UU_U$ with respect to the diamond distance, so it suffices to bound the covering number with respect to $\d$. 
    
    Let $Y_U$ be the set of unitary matrices of the form $UDU^{\dagger}$ with $D$ diagonal of the form $D=(1,e^{i\varphi_2},\dots, e^{i\varphi_d})$, equipped with the operator norm $\norm{\cdot}_{\infty}$.  Let $P: Y_U\to\UU_U$ be the quotient map that takes a unitary matrix to the corresponding unitary channel. We easily check that $P$ is distance-nonincreasing and surjective. Again by \eqref{eq:ncov-P} it suffices to bound the covering number of $Y_U$.
    
   Let $n=\lceil \frac{4\pi}{\vep}\rceil$ and let $N$ be the set of matrices from $Y_U$ with eigenvalue vectors the form $x=(1, x_{i_{2},2}, \dots, x_{i_{d},d})$, where each $i_k=2,\dots, n$ and $x_{j,k} = e^{i \frac{2\pi j}{n}}$. We claim that the set $N$ determines a $\vep/2$-covering of $Y_U$. To this end, for any $V,V' \in Y_U$ we have: 
    \begin{equation}
            \norm{V - V'}_{\infty} = \max_{k\geq 2}\abs{e^{i\varphi'_k} - e^{i\varphi_k}} \leq \max_{k\geq 2}\abs{\varphi'_k - \varphi_k}
    \end{equation}
    By construction of the set $N$ for every $k\geq 2$ the point $\varphi_k$ is $\vep/2$-close on the unit circle to $\varphi'_k = x_{i_k,k}$ for some $i_k$ and this determines the desired channel $V' \in N$ that is $\vep/2$-close to $V$. We thus obtain:
    \begin{equation}
%        \ncov(\UU_U, \vep/2) \leq n^{d-1} =  \left\lceil \frac{4\pi}{\vep}\right\rceil ^{d-1} \leq \left( \frac{8\pi}{\vep} \right)^{d-1}
        \ncov(\UU_U, \vep/2) \leq n^{d-1} =  \left\lceil \frac{4\pi}{\vep}\right\rceil ^{d-1} \leq \left( \frac{4\pi}{\vep}+1 \right)^{d-1}
    \end{equation}
    which finishes the proof.
\end{proof}

\section{Upper bound on recurrence for non-interactive dynamics}

In this part we prove inequality Eq. \eqref{eq:free_recurrence} for arbirtary state $\psi_t$ evolving (from an initial state $\psi_0$)  in $(\mathbb{C}^d)^{\otimes n}$ under a free evolution generated by noninteracting Hamiltonian $H_n$ from Eq.\eqref{eq:nonInteractingHamiltonian}. We first note that:
\begin{equation}
    \exp(-it H_n)=  \overbrace{\exp(-it H)\otimes\ldots\otimes \exp(-it H)}^{n \text{ times}} 
\end{equation}
and hence the evolution operator is a simple tensor product $U_t^{\otimes n}$. If we follow the sketch of the proof of \cref{th:recurrence-main-prl} given in the main text and want to apply the abstract \cref{prop:t+N} we see that:
\begin{equation}\label{eq:npack}
    \trec(\vep)\leq \texit(\vep)\cdot \npack(\mathbb{T}_{\psi_0}^{free},\vep) ,
\end{equation}
where 
\begin{equation}
    \mathbb{T}_{\psi_0}^{free} := \{ \ketbra{\psi} \ \ketbra{\psi} = \mathbf{V}^{\otimes n}[\ketbra{\psi_0}],\ \mathbf{V}\in \mathcal{U}_U \}\ ,
\end{equation}
with $\mathcal{U}_U$ defined above  \cref{lm:packing-UU}, and the packing number in \eqref{eq:npack} is defined with respect to the metric inherited form the trace distance $\dtr$ in the set of pure states. By using $\npack(\mathbb{T}_{\psi_0}^{free},\vep) \leq \ncov(\mathbb{T}_{\psi_0}^{free},\vep)$ it suffices to find an upper bound on the covering number of $\mathbb{T}_{\psi_0}^{free}$. We bound the covering number $\ncov(\mathbb{T}_{\psi_0}^{free},\vep)$ in terms of by a covering number of $\ncov(\mathcal{U}_U,2\vep/n)$.   To this end we note that the map $f(\mathbf{U})\coloneq\mathbf{U}^{\otimes n}[\psi_0]$ satisfies:
\begin{equation}
%    f(\mathbf{U}, \mathbf{V}) \leq \frac{n}{2}\|\U-\V)\|_\diamond ,
   \dtr(f(\mathbf{U}),f(\mathbf{V}) ) \leq \frac{n}{2}\|\U-\V\|_\diamond ,
\end{equation}
which is implied by inequalities:
\begin{equation}
   \dtr( \mathbf{V}^{\otimes n}[\ketbra{\psi_0}], \mathbf{V}^{\otimes n}[\ketbra{\psi_0}]) \leq \frac{1}{2}\|\U^{\otimes n}-\V^{\otimes n}\|_\diamond \leq \frac{n}{2}\|\U-\V\|_\diamond\ ,
\end{equation}
which follow from the variational characterization of the diamond norm and its additivity \cite{NielsenChuang2010}. Therefore, the image of any $2\vep/n$-covering set of $\mathcal{U}_U$ under the map  $f$ is an $\vep$ covering set of $\mathbb{T}_{\psi_0}^{free}$ and henceforth $\ncov(\mathbb{T}_{\psi_0}^{free},\vep)\leq \ncov(\mathcal{U}_U,2\vep/n)$. We conclude the proof of \eqref{eq:free_recurrence} by using $\ncov(\mathcal{U}_U,\vep/n) \leq \left(\frac{4\pi n}{\vep}\right)^{d-1}$ (cf. \cref{lm:packing-UU}).

\section{Lower bound for random Hamiltonians}

\subsection{Lower bound for uniform superpositions}

We will now prove \cref{th:lower-bound-random-prl}. It says that the bound from \cref{th:recurrence-main-prl} is tight in the sense that for generic Hamiltonians the recurrence time is lower bounded by $\sim \texit(\vep)\left(\frac{1}{\vep}\right)^{d-1}$. Therefore it is not possible to have a general upper bound with a smaller functional dependence on $d$ and $\vep$. For concreteness we consider random Hamiltonians $H$ with eigenvalues $\lambda_k, k = 1,\dots, d$ drawn independently uniformly from the interval $[-1,1]$. We take the initial state $\ket{\psi_0}$ to have equal coefficients in all eigenstates of $H$:
\begin{equation}\label{eq:psi0-generic}
        \ket{\psi_0} = \sum_{k=1}^{d}\frac{1}{\sqrt{d}}\ket{k}
\end{equation}
where $\ket{k}$ are the eigenstates of $H$. Everywhere below probabilities and expectations are taken with respect to the choice of $H$. We will prove the following formal version of \cref{th:lower-bound-random-prl}:

\begin{theorem}\label{th:lower-bound-random}
    Let $H$ be a random Hamiltonian with eigenvalues $\lambda_k, k = 1,\dots, d$ as above. Let $\ket{\psi_0}$ be given by \eqref{eq:psi0-generic}. Assume $\vep <\frac{1}{300}$. Then with probability at least $1 - (3/4)^d - 5\exp(-\frac{d}{400})$ the recurrence time for the evolution started at $\ket{\psi_0}$ is at least:
    \begin{equation}
  \trec(\vep) > \left(  \frac{1}{300}\cdot\frac{1}{\vep} \right)^{d-2}
%        \trec(\vep) >  20\left(  \frac{1}{10000}\cdot\frac{1}{\vep}\right)^d
    \end{equation}
    With the same probability the exit time satisfies $\texit(\vep) < 6\vep$, so overall the bound can be stated as:
    \begin{equation}
  \trec(\vep) > \frac{1}{6}\cdot \texit(\vep)\left(  \frac{1}{300}\cdot\frac{1}{\vep} \right)^{d-1} 
    \end{equation}
\end{theorem}

We now outline the proof strategy. To speak of recurrence we need to consider times larger than the exit time $\texit$. We bound it using \cref{th:exit-time-bound-prl}, which requires us to compute the high probability behavior of $\avg{\Delta H^2}, \avg{H^2}$ and $\avg{H^4}$. The time $\texit$ is typically of order $\sim \vep$. Next, we use two different arguments to argue that for any particular time $t > \texit$ the probability that $\dtr(\psi_0, \psi_t) < \vep$ is small. If $t > T$ for some sufficiently large constant $T$, this is done in \cref{th:geometric-bound} using geometric considerations and explicit form of the eigenvalue distribution. For $t$ between $\texit$ and $T$, we argue instead in \cref{prop:200} that the distance from $\psi_0$ is monotonically increasing in this interval, so in particular larger than $\vep$. Finally we discretize time to perform a union bound over finitely many possible times and finish the proof.

\begin{proposition}\label{prop:texit-bound}
    With probability at least $1-5\exp(-\frac{d}{400})$ the following hold simultaneously for $A_1=1/2, A_2=1/9, A_3 = A_1^{-1/2} = \sqrt{2}, A_4= A_1^{-1/2}\cdot (1 - \frac{\vep_0}{A_2})^{-1/2}=2\sqrt{5}, \vep_0 = 1/10$:
    \begin{equation}\label{eq:variance32}
        \avg{\Delta H^2} \leq \avg{H^2} < A_1
    \end{equation}
        
    \begin{equation}\label{eq:114}
            \frac{\avg{\Delta H^2}}{\avg{\Delta H^2} + \sqrt{\avg{(H-\avg{H})^4}}} > A_2
    \end{equation}
    
    \begin{equation}
        \texit(\vep)> A_3 \vep  
    \end{equation}
    If also $\vep < \vep_0$, then simultaneously:
    \begin{equation}\label{eq:texit-upper-bound}
        \texit(\vep) < A_4 \vep
    \end{equation}
\end{proposition}

\begin{proof}
    Since $\avg{H^k} = \frac{1}{d}\sum_{i=1}^{d}\lambda_i^k$ and each $\lambda_i$ is independent and uniform over $[-1,1]$, clearly we have $\E\avg{H} = 0,\E\avg{H^2}= 1/3, \E\avg{H^3}=1/4, \E\avg{H^4} = 1/5$. We bound deviations from these expectations via Hoeffding's inequality.
    \begin{proposition}[Hoeffding's inequality]
Let $X_k$ be independent random variables such that almost surely $a_k \leq X_k \leq b_k$. Let $S_n = X_1 + \dots + X_n$. Then for all $\delta>0$:
\begin{align}\label{eq:Hoeffding}
\Pp\left( S_n - \E S_n \geq \delta \right) \leq \exp\left( -\frac{\delta^2}{\sum_{k=1}^{n}(b_k-a_k)^2} \right)
\end{align}
\end{proposition}
Setting $X_i = \frac{1}{d}\lambda_i^k$ and $\delta=1/10$ we obtain that with probability at least $1-5\exp(-\frac{d}{400})$ simultaneously $\avg{H} < 1/10, 1/5 <\avg{H^2} < 1/2, 0<\avg{H^3}, \avg{H^4} < 1/10$, which also implies $1/10 < \avg{\Delta H^2} < 1/2$. A quick computation shows that these inequalities imply $\avg{(H-\avg{H})^4} < \frac{13}{100}$, so that $\sqrt{\avg{(H-\avg{H})^4}} < 2/5$. Altogether this easily implies \eqref{eq:variance32} and \eqref{eq:114}.

Assuming the same event as above, the quantum speed limit \cite{MandelstamTamm1945} implies:
\begin{equation}
    \texit > \frac{\vep}{\sqrt{\avg{\Delta H^2}}} > \frac{\vep}{A_1^{1/2}} = \sqrt{2}\vep 
    \end{equation}

On the other hand, by \eqref{eq:114}, if $\vep < 1/10$, the conditions of \cref{th:exit-time-bound-prl} are satisfied and one easily calculates that:
    \begin{align}
        &\texit \leq \frac{\vep}{\left( \avg{\Delta H^2} - \left( \avg{\Delta H^2} + \sqrt{\avg{(H-\avg{H})^4}} \right) \vep \right)^{1/2}} = \\
        &\frac{\vep}{\sqrt{\avg{\Delta H^2}}} \cdot \left( 1 - \frac{\avg{\Delta H^2}+\sqrt{\avg{(H-\avg{H})^4}} }{\avg{\Delta H^2}}\vep \right)^{-1/2} < \\
        &\frac{\vep}{A_1^{1/2}} \cdot \left( 1 - \frac{\vep_0}{A_2} \right)^{-1/2} <
        2\sqrt{5} \vep
    \end{align}
\end{proof}

\begin{proposition}\label{prop:200}
    Assume $\vep <  A_2 = \frac{1}{9}$ and let $T = A_2A_3 = \frac{\sqrt{2}}{9}$. Then with probability at least $1-5\exp(-\frac{d}{400})$ we have $\texit < T$ and for every $t \in [\texit, T]$ we have $\dtr(\psi_0, \psi_t) > \vep $.
\end{proposition}

\begin{proof}
By \cref{prop:texit-bound} with probability at least $1-5\exp(-\frac{d}{400})$ we have $\vep_{\ast} > A_2$, where $\vep_{\ast}$ is given as in \cref{th:exit-time-bound-prl}. This follows from the proof of Theorem \ref{th:exit-time-bound-prl}, since by \eqref{eq:derivative} the first derivative of fidelity is always negative for $t < \texit(\vep_{\ast})$ and therefore the distance is monotoneously increasing. Since at time $\texit(\vep)$ the distance is $\vep$ and $\vep < A_2 < \vep_{\ast}$, in the interval $[\texit(\vep), \texit(\vep_{\ast})]$ the distance is $> \vep$. Again by \cref{prop:texit-bound} we have $\texit(\vep_{\ast)} > A_3\vep_{\ast} > A_2 A_3 = \frac{\sqrt{2}}{9}$, which finishes the proof.
\end{proof}

We now prove the main probabilistic bound, which states that for sufficiently large times the state $\psi_t$ is unlikely to be close to the initial state $\psi_0$. Since the trace distance can be expressed using the sum $\frac{1}{d}\sum_k e^{i\lambda_k t}$, which is a sum of independent random variables, it would be natural to again use Hoeffding's inequality. However, this approach can only prove a bound of the form $\sim c^d$ for some constant $c<1$ instead of $\sim \vep^d$ which is needed to get the $\sim(\frac{1}{\vep})^d$ dependence of the recurrence time for arbitrary $\vep$. We thus use a more geometric approach which uses directly the distribution of eigenvalues.

\begin{proposition}\label{th:geometric-bound}
 Let $T = A_2A_3=\frac{\sqrt{2}}{9}$. Then for any $t>T$ we have:
    \begin{equation}\label{eq:bound2}
        \Pr\left( \dtr(\psi_0, \psi_t) < \vep \right) \leq  (80\vep)^{d-1}
    \end{equation}
\end{proposition}

\begin{proof}
   Let us work with the Bures distance $\dbures$. By \eqref{eq:dtr-hs} if $\dtr(\psi_0, \psi_t) < \vep$, then $\dbures(\psi_0,\psi_t) < \sqrt{2}\vep$. On the other hand:
    \begin{equation}\label{eq:bures-bound}
        2\vep^2 > \dbures(\psi_0,\psi_t)^2 = 2(1 - \abs{\langle \psi_0 \vert \psi_t \rangle}) = 
        2\left(1 - \abs{\frac{1}{d}\sum_{k=1}^{d}e^{i\lambda_k t}}\right)
    \end{equation}
    Let $\alpha$ be the phase of $\sum_{k}e^{i\lambda_k t}$. We then have by \eqref{eq:bures-bound}:
    \begin{equation}
        \sum_{k=1}^{d}\abs{e^{i\lambda_k t} - e^{i\alpha} }^2 = 2d - 2\mathrm{Re}\left(e^{-i\alpha} \sum_k e^{i\lambda_k t} \right) = 
        2d - 2\abs{\sum_k e^{i\lambda_k t}} < 2d \vep^2
    \end{equation}
    Therefore there exists an index $j$ such that:
    \begin{equation}\label{eq:indexj}
    \abs{e^{i\lambda_j t}-e^{i\alpha}}< \sqrt{2}\vep < 2\vep
    \end{equation}
     We will later perform a union bound over $j$. For $j\neq k$, by using $\abs{e^{i\lambda_k t} - e^{i\lambda_j t}} \leq \abs{e^{i\lambda_k t} - e^{i\alpha}} + \abs{e^{i\alpha} - e^{i\lambda_j t}}$ we easily arrive at:
     \begin{equation}
         \sum_{k \neq j}\abs{e^{i\lambda_k t} - e^{i\lambda_j t}}^2 \leq  2\sum_{k\neq j}\abs{e^{i\lambda_k t} - e^{i\alpha} }^2 + 2(d-1)\abs{e^{i\alpha} - e^{i\lambda_j t}}^2
     \end{equation}
     and the previous bounds then imply that:
     \begin{equation}
         \sum_{k \neq j}\abs{e^{i\lambda_k t} - e^{i\lambda_j t}}^2 < 8d\vep^2 + 8(d-1)\vep^2 < 16d\vep^2
     \end{equation}
     which is equivalent to:
     \begin{equation}
         \sum_{k \neq j}\abs{1 - e^{i(\lambda_k - \lambda_j) t}}^2 <  16d\vep^2
     \end{equation}
     In order to use the explicit form of the eigenvalue distribution we pass to the angular distance by using $\abs{1 - e^{i\varphi}} \geq \frac{2}{\pi}\abs{\varphi}$, valid for $\varphi \in [-\pi,\pi]$. Assume for notational convenience that $j=d$. Let $z_k := (\lambda_k-\lambda_j)t \mod 2\pi$. The above inequalities then imply that the point $(z_1,\dots,z_{d-1})$ on the $d-1$-dimensional torus $\mathbb{T}^{d-1} = \mathbb{R}^{d-1} / (2\pi \mathbb{Z})^{d-1}$ lies within distance $\frac{\pi}{2}\sqrt{d}\vep$ of $0$, where distance is measured by $(\sum_k\abs{z_k})^{1/2}$.   

     Let $f=f_2 \circ f_1$, where $f_1(\lambda_1,\dots,\lambda_d) = t\cdot(\lambda_1 - \lambda_d,\dots,\lambda_{d-1}-\lambda_d)$ and $f_2(x_1,\dots,x_{d-1}) = (x_1,\dots,x_{d-1}) \mod 2\pi$. Let $g$ denote the density of the probability distribution which is the image of the uniform distribution over $[-1,1]^d$ under the map $f_1$. Likewise, let $h$ denote the density of the image under the composite map $f$. By the above discussion it suffices to bound the probability of the ball $B_{\mathbb{T}^{d-1}}(0,r)$ for $r = \frac{\pi}{2}\sqrt{d}\vep$ under the probability density $h$.
     
      We first claim that the probability density $g$ is pointwise bounded by $(2t)^{-(d-1)}$. To see this, condition on the value $\lambda_d = \lambda$ and compute the density for $t=1$:
     \begin{equation}
         g(z_1,\dots, z_{d-1}) = \frac{1}{2^d}\int_{\R} \id_{[-1,1]}(\lambda) \prod_{i=1}^{d-1}\id_{[-1,1]}(\lambda + z_i) d\lambda
     \end{equation}
     It is easy to see that the support of the integrand is always an interval of length at most $2$ and hence $g \leq 2^{-(d-1)}$ as claimed. The bound $(2t)^{-(d-1)}$ follows by scaling.

     To obtain the periodized density $h$, we need to sum over possible preimages of a point:
     \begin{equation}
         h(\varphi_1,\dots,\varphi_{d-1}) = \sum_{m \in \Z^{d-1}}g(\varphi_1+2\pi m_1,\dots,\varphi_{d-1}+2\pi m_{d-1})
     \end{equation}
     The distribution of each coordinate on the right hand side is supported on the interval $[-2t,2t]$ and therefore a point $m \in \Z^{d-1}$ can give a nonzero contribution only if $\varphi_i + 2\pi m_i \in [-2t,2t]$ for each $i$. The number of possible integers that satisfy such a bound for a single coordinate is at most $1 + \frac{2t}{\pi}$. Hence the overall bound becomes:
     \begin{equation}
         h(z_1,\dots, z_{d-1}) \leq \left(1 + \frac{2t}{\pi}\right)^{d-1}\cdot (2t)^{-(d-1)} = \left(\frac{1}{2t} + \frac{1}{\pi}\right)^{-(d-1)} \leq \left(\frac{1}{2T} + \frac{1}{\pi}\right)^{-(d-1)} < 4^{-(d-1)}
     \end{equation}
     It now suffices to integrate the above density bound over the ball $B_{\mathbb{T}^{d-1}}(0,r)$. The volume of this ball is obviously at most the volume of the same ball in $\R^{d-1}$, since the projection of $\R^{d-1}$ onto $\mathbb{T}^{d-1}$ can only identify points and thus decrease the volume. We thus have:
     \begin{equation}
         \int_{B_{\mathbb{T}^{d-1}}(0,r)}h(z)dz \leq \mathrm{vol}(B_{\R^{d-1}}(0,r)) \cdot 4^{-(d-1)}
     \end{equation}
     It can be easily computed using Stirling approximation that $\vol_{\R^{d-1}}(B(0,r)) \leq \left(\frac{\pi}{2}\vep\right)^{d-1}(\sqrt{2e\pi})^{d-1} < (10\vep)^{d-1}$ and therefore:
     \begin{equation}
         \int_{B_{\mathbb{T}^{d-1}}(0,r)}h(z)dz < (40 \vep)^{-(d-1)}        
     \end{equation}
    Finally, recall that we need to perform a union bound over possible values of the index $j$ from \eqref{eq:indexj}. This gives a factor of $d$, which we can bound by $d \leq 2^{d-1}$ and thus overall obtain:
    \begin{equation}
    \Pr\left( \dtr(\psi_0, \psi_t) < \vep \right) \leq  (80\vep)^{d-1}
    \end{equation}
\end{proof}

\begin{proof}[Proof of \cref{th:lower-bound-random}]
    We will prove that with probability $1 - (3/4)^d - 5\exp(-\frac{d}{400})$ for every $t \in [\texit, T_0]$, for $T_0$ to be chosen later, we have $\dtr(\psi_0,\psi_t) > \vep$. Since by \cref{prop:texit-bound} we have $\texit < A_4\vep = 2\sqrt{5}\vep < T_0$ if $T_0$ is sufficiently large, this  will prove that the time of first recurrence is at least $T_0$.

    By \cref{prop:200} we automatically have $\dtr(\psi_0,\psi_t) > \vep$ for all $t \in [\texit, T]$ with probability at least $1 - 5\exp(-\frac{d}{400})$. This means it suffices to prove the statement for all times $t \in [T, T_0]$.

    Let us divide the interval $[T, T_0]$ into points $t_i, i=1,\dots,n$, where $t_{i+1}-t_{i} = \vep $ and $n<\lceil \frac{T_0}{\vep} \rceil$. We perform a union bound over $t_i$:
    \begin{equation} 
        \Pr(\exists_i \dtr(\psi_0, \psi_{t_i}) < 2\vep ) \leq \sum_{i=1}^{n} \Pr(\dtr(\psi_0, \psi_{t_i}) < 2\vep ).
    \end{equation}
    
Since $t_i >  T$ the bound \eqref{eq:bound2} holds and there are $\leq \frac{T_0}{\vep}$ such points, so:
    \begin{equation}
%    \sum_{t_i \in [T_1, T]}\Pr(\dtr(\psi_0, \psi_{t_i}) < 2\vep ) \leq \frac{T_0}{\vep}\cdot\frac{2\pi}{\vep^2} \left( C\vep\right)^d
    \sum_{t_i \in [T_1, T]}\Pr(\dtr(\psi_0, \psi_{t_i}) < 2\vep ) \leq \frac{T_0}{\vep}\cdot  \left( 80\vep\right)^{d-1}
    \end{equation}

Now, for any $t \in [t_{i},t_{i+1}]$ we use the triangle inequality:
    \begin{equation}
        \dtr(\psi_0, \psi_{t}) \geq \dtr(\psi_0, \psi_{t_{i}}) - \dtr(\psi_{t_{i}}, \psi_{t}) 
    \end{equation}
    Noticing that $\dtr(\psi_t, \psi_{t_{i}}) = \dtr(\psi_0, \psi_{t-t_{i}})$ and using \eqref{eq:dtr-hs} it is straightforward to show that $\dtr$ has Lipschitz constant at most $\sqrt{\frac{1}{d}\sum_{k=1}^{d} \lambda_k^2}$, that is:
    \begin{equation}
        \dtr(\psi_{t_{i}}, \psi_{t}) \leq \sqrt{\frac{1}{d}\sum_{k=1}^{d} \lambda_k^2} \cdot \abs{t-t_i}
    \end{equation}
    Since by definition we have $\abs{t-t_i} \leq \vep / 2$ and by \cref{prop:texit-bound} $\sqrt{\frac{1}{d}\sum_{k=1}^{d} \lambda_k^2} = \sqrt{\avg{H^2}} < \sqrt{1/2} < 1$, in the end we have:
    \begin{equation}
    \dtr(\psi_0, \psi_{t}) \geq \dtr(\psi_0, \psi_{t_{i}}) - \dtr(\psi_{t_{i}}, \psi_{t}) >
    2\vep - \vep = \vep
    \end{equation}
    Overall we obtain the following bound:
    \begin{equation}\label{eq:final-bound}
  %      \Pr(\exists_t \in [\texit, T_0] \dtr(\psi_0, \psi_{t_i}) < \vep ) \leq
     %   \frac{T_0}{\vep}\cdot\frac{2\pi}{\vep^2} \left( C\vep\right)^d
              \Pr(\exists_t \in [\texit, T_0] \dtr(\psi_0, \psi_{t_i}) < \vep ) \leq
        \frac{T_0}{\vep}\cdot\left( 80\vep\right)^{d-1}
    \end{equation}
    Now it suffices to choose $T_0$ such that the right hand side is at most $(3/4)^d$, so it suffices to choose:
    \begin{equation}
%        T_0 = \frac{\vep^3}{2\pi}\left(  \frac{3}{4C}\cdot\frac{1}{\vep} \right)^d > 20\left(\frac{1}{10000\vep}\right)^d
   T_0 = \left(  \frac{3}{4}\right)^d \left(\frac{1}{80}\right)^{d-1} \vep^{-(d-2)}   > \left(\frac{1}{300\vep}\right)^{d-2}
    \end{equation}.
\end{proof}

\subsection{Lower bound for states with small effective dimension}

As mentioned in the main text, the scaling of the upper bound on recurrence time can be improved from $\sim(1/\vep)^d$ if the initial state $\ket{\psi_0}$ has effective support smaller than $d$. It is natural to ask what property of the initial state and the Hamiltonian governs the true behavior of the recurrence time. One possible candidate is the effective dimension, also known as the inverse participation ratio, defined as:
\begin{equation}
    d_{\mathrm{eff}} := \frac{1}{\sum_{i=1}^{d}\abs{a_i}^4}
\end{equation}
The effective dimension has been used together with concentration techniques e.g. in \cite{alvaro-recurrence} to give lower bounds on average recurrence times for certain observables. It is reasonable to ask whether this reflects the true behavior of the recurrence time, i.e. if it is possible to prove a matching upper bound depending on $d_{\mathrm{eff}}$ instead of $d$. Below we refute this possibility for random Hamiltonians -- with high probability, we give an example of an initial state which has low effective dimension, yet whose recurrence time  scales exponentially in $d$ instead of $d_{\mathrm{eff}}$.

Let $H$ be as before. Fix $\eta \in [0,1]$ and let:
 \begin{equation}\label{eq:state-eta}
     \ket{\psi_0} = \sum_{k=1}^{d}a_k\ket{k}
 \end{equation}
with $a_1 = \sqrt{\eta}$ and $a_k = \sqrt{\frac{1-\eta}{d-1}}$ for $k=2,\dots,d$, with the ordering of the eigenstates chosen randomly. To get a useful probabilistic bound we will require that $\eta = o(1)$ as $d \to \infty$. For technical simplicity we also require $\eta \gg \frac{1}{d}$.

We easily compute that:
\begin{equation}\label{eq:deff}
%        d_{\mathrm{eff}} = \left( \sum_{i=1}^{d}\abs{a_i}^4 \right)^{-1} = \left( \eta^2 + \frac{(1-\eta)^2}{d-1}  \right)^{-1} =
   %     \frac{d-1}{d\eta^2 + 1 -2\eta} \sim d^{1/2} \ll d
                d_{\mathrm{eff}} = \left( \sum_{i=1}^{d}\abs{a_i}^4 \right)^{-1} = \left( \eta^2 + \frac{(1-\eta)^2}{d-1}  \right)^{-1} =
        \frac{d-1}{d\eta^2 + 1 -2\eta} < \frac{1}{\eta^2 - \frac{2\eta}{d}} \sim \frac{1}{\eta^2}
\end{equation}
since by assumption $\eta \gg \frac{1}{d}$.

\begin{theorem}\label{th:lower-bound-random-eta}
    Let $H$ be a random Hamiltonian with eigenvalues $\lambda_k, k = 1,\dots, d$ as above. Let $\ket{\psi_0}$ be given by \eqref{eq:state-eta}.  Assume $\vep <\frac{1}{600}$. Also assume that $\vep > \sqrt{\eta}$. Then with probability at least $1 - (3/4)^d - 5\exp(-\frac{d_{\mathrm{eff}}}{400})$, with $d_{\mathrm{eff}}$ given by \eqref{eq:deff}, the recurrence time for the evolution started at $\ket{\psi_0}$ is at least:
    \begin{equation}
%        \trec(\vep) >   20\left(  \frac{1}{20000}\cdot\frac{1}{\vep}\right)^{d-1} \gg \left(\frac{1}{\vep}\right)^{d_{\mathrm{eff}}}
        %\trec(\vep) >   20\left(  \frac{1}{20000}\cdot\frac{1}{\vep}\right)^{d-1} 
         \trec(\vep) > \left(  \frac{1}{600}\cdot\frac{1}{\vep} \right)^{d-3}
    \end{equation}
\end{theorem}

The above theorem provides a family of examples where $d_{\mathrm{eff}} \ll d$, yet the effective support is $\sim d$ and the recurrence time is lower bounded by $\sim(\frac{1}{\vep})^d \gg (\frac{1}{\vep})^{d_{\mathrm{eff}}}$ with probability $1-o(1)$. For concreteness we can take e.g. $\eta = \frac{1}{\sqrt{\log d}}$ to obtain $d_{\mathrm{eff}} \sim \log d \ll d$. Then the above recurrence lower bound holds with probability $1-o(1)$ whenever $\vep > (\log d)^{-1/4}$

\begin{proposition}\label{prop:eff-support-hk}
    With probability at least $1-5\exp(-\frac{d_{\mathrm{eff}}}{400})$ the same bounds as in \cref{prop:texit-bound} hold.
\end{proposition}

    \begin{proof}
The computation of $\E\avg{H^k}$ is identical since these expected values with respect to the choice of $H$ do not depend on the state, only on the distribution of eigenvalues. Similarly to \cref{prop:texit-bound} we use Hoeffding's equality to bound the deviations of $\avg{H^k}, k=1,2,3,4$ from their respective expected values. We set $X_i = \abs{a_i}^2\lambda_i^4$ and $\delta=1/10$. Note that this time we have $-\abs{a_i}^2 \leq X_i \leq \abs{a_i}^2$, so the right hand side of \eqref{eq:Hoeffding} becomes:
\begin{equation}
    \Pp\left( S_n - \E S_n \geq \delta \right) \leq  
\exp\left( -\frac{\delta^2}{4\sum_{k=1}^{n}\abs{a_k}^4} \right) = 
\exp\left( -\frac{d_{\mathrm{eff}}}{400}\right) 
\end{equation}
which finishes the proof analogously to \cref{prop:texit-bound}
    \end{proof}

\begin{proposition}\label{prop:200-eta}
    Assume $\vep <  A_2 = \frac{1}{9}$ and let $T = A_2A_3 = \frac{\sqrt{2}}{9}$. Then with probability at least $1-5\exp(-\frac{d_{\mathrm{eff}}}{400})$ we have $\texit < T$ for every $t \in [\texit, T]$ we have $\dtr(\psi_0, \psi_t) > \vep $.
\end{proposition}

\begin{proof}
    Identical to the proof of \cref{prop:200}, only using \cref{prop:eff-support-hk} instead of \cref{prop:texit-bound}. The only difference between these two Propositions is the success probability $1-5\exp(-\frac{d_{\mathrm{eff}}}{400})$ instead of $1-5\exp(-\frac{d}{400})$.
\end{proof}

\begin{proposition}\label{prop:eta-bound}
Suppose that $\vep > \sqrt{\eta}$. Then for any fixed $t > T=\frac{\sqrt{2}}{9}$ we have:
    \begin{equation}
         %\Pr\left( \dtr(\psi_0, \psi_t) < \vep \right) \leq  4 \exp\left( -\frac{\delta_0^2 (d-1)}{32} \right)
        \Pr\left( \dtr(\psi_0, \psi_t) < \vep \right) \leq  (160\vep)^{d-2}
    \end{equation}
\end{proposition}

\begin{proof}
Suppose that $\dtr(\psi_0,\psi_t) < \vep$ for some $t$. Since:
\begin{equation}
    \vep^2 > \dtr(\psi_0,\psi_t)^2 = 1 - \abs{\sum_{k=1}^{d}\abs{a_k}^2 e^{i\lambda_k t} }^2 = 
    1 - \abs{\eta e^{i\lambda_1 t} + \frac{1-\eta}{d-1}\sum_{k=2}^{d} e^{i\lambda_k t} }^2
\end{equation}
so that:
\begin{equation}
    \abs{\eta e^{i\lambda_1 t} + \frac{1-\eta}{d-1}\sum_{k=2}^{d} e^{i\lambda_k t} } > \sqrt{1-\vep^2} 
\end{equation}
By the triangle inequality this easily implies:
\begin{equation}\label{eq:d-1_sum}
    \abs{\frac{1}{d-1}\sum_{k=2}^{d} e^{i\lambda_k t} } > \frac{\sqrt{1-\vep^2}-\eta}{1-\eta} > \sqrt{1-\vep^2}-\eta 
\end{equation}

Since we assume $\vep > \sqrt{\eta}$, this implies $\eta < \vep^2$, which in turn implies that $\sqrt{1-\vep^2}-\eta > \sqrt{1-4\vep^2}$. Thus the event \eqref{eq:d-1_sum} implies the event:
\begin{equation}
     \dtr(\phi_0,\phi_t) < 2\vep
\end{equation}
where $\ket{\phi_t} = \frac{1}{\sqrt{d-1}}\sum_{k=2}^{d}e^{i\lambda_k t}\ket{k}$. We are now in position to use \cref{th:geometric-bound}, only with state $\phi_t$ of dimension $d-1$ instead of $d$ and setting $\vep\to 2\vep$, and the claimed bound follows.
\end{proof}

\begin{proof}[Proof of \cref{th:lower-bound-random-eta}]
    Follows from repeating the proof of \cref{th:lower-bound-random} verbatim, only using \cref{prop:200-eta} and \cref{prop:eta-bound} when appropriate instead of \cref{prop:200} and \cref{th:geometric-bound}.
\end{proof}

\section{Miscellaneous remarks on exit time for states}\label{sec:states}

There seems to be no simple criterion for finiteness of $\texit$ in full generality. However, the following proposition gives a sufficient and necessary condition under mild assumptions on $H$. Namely, we require that the eigenvalues of $H$ are rationally independent, i.e. if:
\begin{equation}
    \sum_{k=1}^{d} q_k \lambda_k = 0
\end{equation}
for $q_k \in \mathbb{Q}$, then $q_k=0$ for all $k$. This assumption will be satisfied by a generic Hamiltonian, i.e. any Hamiltonian $H$ may be perturbed by an arbitrarily small amount to a Hamiltonian $H'$ whose eigenvalues are rationally independent.

%(CITE MATH SE GUY)

\begin{proposition}
    Suppose that all eigenvalues of $H$ are rationally independent. Then $\texit$ is finite if and only if:
    \begin{equation}
         1 - \vep^2 \geq \max\{2 \max_k\abs{a_k}^2-1, 0\}
    \end{equation}
\end{proposition}

\begin{proof}
    We have:
    \begin{equation}
        \ket{\psi_t} = \sum_{k=1}^{d}a_k e^{i\lambda_k t}\ket{k}
    \end{equation}
    so that:
    \begin{equation}\label{eq:dtr-exit}
        \dtr(\psi_0, \psi_t)^2 = 1 - \abs{\ketscalar{\psi_0}{\psi_t}}^2 = 1 - \abs{ \sum_{k=1}^{d}\abs{a_k}^2 e^{i\lambda_k t}}^2
    \end{equation}
    The exit time $\texit$ is finite if for some $t$ we have $\dtr(\psi_0, \psi_t) \geq \vep$, which by \eqref{eq:dtr-exit} is equivalent to:
    \begin{equation}
        \abs{ \sum_{k=1}^{d}\abs{a_k}^2 e^{i\lambda_k t}}^2 \leq 1 - \vep^2
    \end{equation}
    Let us consider:
    \begin{equation}\label{eq:infimum}
        M(H,\psi_0) := \inf_{t \geq 0} \abs{ \sum_{k=1}^{d}\abs{a_k}^2 e^{i\lambda_k t}}^2
    \end{equation}
    By assumption, the eigenvalues of $H$ are rationally independent, which implies that the trajectory $\{(e^{i\lambda_1 t},\dots,e^{i\lambda_d t})\}_{t \geq 0}$ is dense on the $d$-dimensional torus (this is known as Kronecker's theorem). By continuity, this implies that \eqref{eq:infimum} can be equivalently written as:
    \begin{equation}\label{eq:infimum-torus}
        M(H,\psi_0) := \inf_{\varphi_1, \dots, \varphi_d \in [0,2\pi)} \abs{ \sum_{k=1}^{d}\abs{a_k}^2 e^{i\varphi_k}}^2
    \end{equation}
    For convenience we write $p_k = \abs{a_k}^2$. Let us sort the coefficients so that $p_1$ is the largest. We consider two cases.

    The first case is $p_1 \leq \frac{1}{2}$. In that case the infimum \eqref{eq:infimum-torus} is actually equal to zero. To see this, we need to find $\varphi_k$ so that the complex numbers $z_k = p_k e^{i\varphi_k}$ sum to zero, i.e. comprise sides of a polygon in the complex plane. It is easy to prove that is is always possible by induction on $d$ (the base case $d=3$ corresponds to creating a triangle with side lengths $p_1, p_2, p_3$, which is possible since all $p_k$ are at most $1/2$.).

    The other case is $p_1 > \frac{1}{2}$. We claim that the infimum is attained when $e^{i\varphi_1} =1 $ and $e^{i\varphi_k} = -1$ for $k=2,\dots,d$. Clearly for such a choice of phases we obtain the value $(2p_1 -1)^2$, which is $>0$ since we assumed $p_1 > \frac{1}{2}$. To see that this is optimal, we use the triangle inequality:
    \begin{equation}
        \abs{ \sum_{k=1}^{d} p_k e^{i\varphi_k} } \geq \abs{p_1 e^{i\varphi_1}} - \abs{\sum_{k=2}^{d} p_k e^{i\varphi_k} } \geq \abs{p_1 e^{i\varphi_1}} - \sum_{k=2}^{d}\abs{ p_k e^{i\varphi_k} } = p_1 - \sum_{k=2}^{d}p_k = 2p_1 - 1
    \end{equation}

    Combining the two cases together, we obtain the condition:
    \begin{equation}
         1 - \vep^2 \geq M(H,\psi_0) = \max\{2 \abs{a_1}^2-1, 0\}
    \end{equation}
    which finishes the proof.
\end{proof}

\end{document}